\pdfoutput=1
\documentclass[11pt]{article}
\usepackage{amsmath}
\usepackage{amsthm}
\usepackage{mathptmx}
\usepackage[libertine]{newtxmath}
\usepackage[T1]{fontenc}
\usepackage{geometry}
\usepackage{array}
\usepackage{pifont}
\usepackage{float}
\usepackage{graphicx}
\usepackage{nameref}

\makeatletter

\providecommand{\tabularnewline}{\\}

\usepackage{etoolbox}
\usepackage[bottom,hang]{footmisc}
\usepackage[medium,compact]{titlesec}	
\usepackage{xcolor}
\usepackage{hyperref}	
\hypersetup
{
  unicode=true,
  colorlinks=true,
  linkcolor=[rgb]{0 0 .8},
  citecolor=[rgb]{0 .8 0},
  urlcolor=[rgb]{.8 0 .8},
  breaklinks=true,  
  pdfstartview=,
  bookmarksopen=true,
  pdfauthor={Lim Wei Quan}
}
\usepackage[numbered]{bookmark}
\usepackage[all]{hypcap}
\usepackage[nameinlink]{cleveref}	

\AtBeginDocument{\let\ref\cref}	
\crefname{section}{Section}{Sections}
\crefname{figure}{Figure}{Figures}
\crefname{table}{Table}{Tables}
\crefname{example}{Example}{Examples}
\crefname{footnote}{}{}
\crefformat{footnote}{#2\footnotemark[#1]#3}
\makeatletter
\patchcmd\@combinedblfloats
  {\box\@outputbox}
  {\unvbox\@outputbox}
  {}{\errmessage{\noexpand\@combinedblfloats could not be patched}}
\makeatother
\makeatletter\g@addto@macro{\UrlBreaks}{\UrlOrds}\makeatother	
\usepackage{graphicx}
\let\oldsqrt\sqrt
\renewcommand{\sqrt}[2][\ \,\,]{{\!\!\oldsqrt[\raisebox{.1em}{\scalebox{.7}{$#1$}}]{#2}\,}}
  \let\originalleft\left
  \let\originalright\right
  \renewcommand{\left}{\mathopen{}\mathclose\bgroup\originalleft}
  \renewcommand{\right}{\aftergroup\egroup\originalright}
\binoppenalty=\maxdimen
\DeclareMathSizes{5}{5}{4}{3}
\DeclareMathSizes{6}{6}{5}{4}
\DeclareMathSizes{7}{7}{5.5}{4}
\DeclareMathSizes{8}{8}{6.5}{5}
\DeclareMathSizes{9}{9}{7.5}{6}
\DeclareMathSizes{10}{10}{8.5}{7}
\DeclareMathSizes{10.95}{10.95}{9}{7}	
\DeclareMathSizes{12}{12}{10}{8}
\DeclareMathSizes{14.4}{14.4}{12}{10}
\DeclareMathSizes{17.28}{17.28}{14.5}{11.5}
\DeclareMathSizes{20.74}{20.74}{17}{13.5}
\DeclareMathSizes{24.88}{24.88}{20}{16}
\makeatletter
\AtBeginDocument
{
  \check@mathfonts
  \fontdimen13\textfont2 = \fontdimen13\textfont2		
  \fontdimen14\textfont2 = \fontdimen13\textfont2		
  \fontdimen15\textfont2 = \fontdimen13\textfont2		
  \fontdimen16\textfont2 = 1.3\fontdimen16\textfont2	
  \fontdimen17\textfont2 = \fontdimen16\textfont2		
  \fontdimen14\scriptfont2 = \fontdimen13\scriptfont2	
  \fontdimen15\scriptfont2 = \fontdimen13\scriptfont2
  \fontdimen16\scriptfont2 = 1.3\fontdimen16\scriptfont2
  \fontdimen17\scriptfont2 = \fontdimen16\scriptfont2
  \fontdimen14\scriptscriptfont2 = \fontdimen13\scriptscriptfont2	
  \fontdimen15\scriptscriptfont2 = \fontdimen13\scriptscriptfont2
  \fontdimen16\scriptscriptfont2 = 1.3\fontdimen16\scriptscriptfont2
  \fontdimen17\scriptscriptfont2 = \fontdimen16\scriptscriptfont2
}
\makeatother
\usepackage[skip=.5em]{caption}
\newcommand{\setmulength}[2]{#1=#2\relax}	
\setmulength{\thinmuskip}{2mu plus 1mu minus 1mu}
\setmulength{\medmuskip}{2mu plus 1mu minus 1mu}
\setmulength{\thickmuskip}{4mu plus 1mu minus 2mu}
\AtBeginEnvironment{pmatrix}
{
  \setlength{\arraycolsep}{.4em}
  
}
\usepackage{microtype}
\DisableLigatures{encoding=OT1,family=ntxtlf}
\DeclareSymbolFont{unspacedletters}{OT1}{ntxtlf}{m}{it}
\SetSymbolFont{unspacedletters}{bold}{OT1}{ntxtlf}{b}{it}
\DeclareMathSymbol{A}{\mathalpha}{unspacedletters}{`A}
\DeclareMathSymbol{B}{\mathalpha}{unspacedletters}{`B}
\DeclareMathSymbol{C}{\mathalpha}{unspacedletters}{`C}
\DeclareMathSymbol{D}{\mathalpha}{unspacedletters}{`D}
\DeclareMathSymbol{E}{\mathalpha}{unspacedletters}{`E}
\DeclareMathSymbol{F}{\mathalpha}{unspacedletters}{`F}
\DeclareMathSymbol{G}{\mathalpha}{unspacedletters}{`G}
\DeclareMathSymbol{H}{\mathalpha}{unspacedletters}{`H}
\DeclareMathSymbol{I}{\mathalpha}{unspacedletters}{`I}
\DeclareMathSymbol{J}{\mathalpha}{unspacedletters}{`J}
\DeclareMathSymbol{K}{\mathalpha}{unspacedletters}{`K}
\DeclareMathSymbol{L}{\mathalpha}{unspacedletters}{`L}
\DeclareMathSymbol{M}{\mathalpha}{unspacedletters}{`M}
\DeclareMathSymbol{N}{\mathalpha}{unspacedletters}{`N}
\DeclareMathSymbol{O}{\mathalpha}{unspacedletters}{`O}
\DeclareMathSymbol{P}{\mathalpha}{unspacedletters}{`P}
\DeclareMathSymbol{Q}{\mathalpha}{unspacedletters}{`Q}
\DeclareMathSymbol{R}{\mathalpha}{unspacedletters}{`R}
\DeclareMathSymbol{S}{\mathalpha}{unspacedletters}{`S}
\DeclareMathSymbol{T}{\mathalpha}{unspacedletters}{`T}
\DeclareMathSymbol{U}{\mathalpha}{unspacedletters}{`U}
\DeclareMathSymbol{V}{\mathalpha}{unspacedletters}{`V}
\DeclareMathSymbol{W}{\mathalpha}{unspacedletters}{`W}
\DeclareMathSymbol{X}{\mathalpha}{unspacedletters}{`X}
\DeclareMathSymbol{Y}{\mathalpha}{unspacedletters}{`Y}
\DeclareMathSymbol{Z}{\mathalpha}{unspacedletters}{`Z}
\DeclareMathSymbol{a}{\mathalpha}{unspacedletters}{`a}
\DeclareMathSymbol{b}{\mathalpha}{unspacedletters}{`b}
\DeclareMathSymbol{c}{\mathalpha}{unspacedletters}{`c}
\DeclareMathSymbol{d}{\mathalpha}{unspacedletters}{`d}
\DeclareMathSymbol{e}{\mathalpha}{unspacedletters}{`e}
\DeclareMathSymbol{f}{\mathalpha}{unspacedletters}{`f}
\DeclareMathSymbol{g}{\mathalpha}{unspacedletters}{`g}
\DeclareMathSymbol{h}{\mathalpha}{unspacedletters}{`h}
\DeclareMathSymbol{i}{\mathalpha}{unspacedletters}{`i}
\DeclareMathSymbol{j}{\mathalpha}{unspacedletters}{`j}
\DeclareMathSymbol{k}{\mathalpha}{unspacedletters}{`k}
\DeclareMathSymbol{l}{\mathalpha}{unspacedletters}{`l}
\DeclareMathSymbol{m}{\mathalpha}{unspacedletters}{`m}
\DeclareMathSymbol{n}{\mathalpha}{unspacedletters}{`n}
\DeclareMathSymbol{o}{\mathalpha}{unspacedletters}{`o}
\DeclareMathSymbol{p}{\mathalpha}{unspacedletters}{`p}
\DeclareMathSymbol{q}{\mathalpha}{unspacedletters}{`q}
\DeclareMathSymbol{r}{\mathalpha}{unspacedletters}{`r}
\DeclareMathSymbol{s}{\mathalpha}{unspacedletters}{`s}
\DeclareMathSymbol{t}{\mathalpha}{unspacedletters}{`t}
\DeclareMathSymbol{u}{\mathalpha}{unspacedletters}{`u}
\DeclareMathSymbol{v}{\mathalpha}{unspacedletters}{`v}
\DeclareMathSymbol{w}{\mathalpha}{unspacedletters}{`w}
\DeclareMathSymbol{x}{\mathalpha}{unspacedletters}{`x}
\DeclareMathSymbol{y}{\mathalpha}{unspacedletters}{`y}
\DeclareMathSymbol{z}{\mathalpha}{unspacedletters}{`z}
\pretocmd{\section}{\addvspace{0em plus 2em}\addpenalty{-2000}}{}{}
\pretocmd{\subsection}{\addvspace{0em plus 1.5em}\addpenalty{-1500}}{}{}
\pretocmd{\subsubsection}{\addvspace{0em plus 1em}\addpenalty{-1000}}{}{}
\pretocmd{\par}{\addpenalty{-1000}}{}{}
  \usepackage{enumitem}
  \setlistdepth{14}
  \makeatletter
  \renewcommand\p@enumii{\theenumi}
  \renewcommand\p@enumiii{\theenumi\theenumii}
  \renewcommand\p@enumiv{\theenumi\theenumii\theenumiii}
  \makeatother
  \newlength{\listspace}
  \setlist{topsep=\listspace,itemsep=\listspace,parsep=0em,partopsep=0em}
  \setlist[enumerate]{leftmargin=2em}
  \setlist[itemize]{leftmargin=1.5em}
  \AtBeginEnvironment{itemize}{\interlinepenalty=3000}
  \AtBeginEnvironment{enumerate}{\interlinepenalty=3000}
\makeatletter\newcommand{\justified}{\rightskip\z@skip\leftskip\z@skip}\makeatother
\date{}
\usepackage{upgreek}
\usepackage{calc}	
\newlength{\linespace}
\newcommand\makelinespace{\setlength{\linespace}{\baselineskip-1em}\vspace{\linespace}}
\newlength{\parspace}
\newcommand\makeparspace{\setlength{\parspace}{\parskip+\baselineskip-1em}\vspace{\parspace}}
\newcommand{\forcefontspace}{\par\vspace{-\baselineskip}\vphantom{ABCDEgjpqy}}	
\let\incgraphics\includegraphics
\newsavebox{\imagebox}
\newlength{\imagerule}
\newcommand{\imagescale}{1}
\newcommand{\scalegraphics}[1]{\renewcommand{\imagescale}{#1}}
\renewcommand{\includegraphics}[2][]
{%
\def\image{\scalebox{\imagescale}{\incgraphics[#1]{#2}}}%
\savebox{\imagebox}{\image}%
\setlength{\imagerule}{\ht\imagebox+\baselineskip-.7em}%
\ifvmode{\forcefontspace}\fi\rule[0em]{0em}{\imagerule}\image%
}
\let\oldfigure\figure
\let\oldtable\table
\makeatletter
\def\beginfloat{\centering\vspace{.1em}\makeparspace}
\def\figure@i[#1]{\oldfigure[#1]\beginfloat}
\def\figure@ii{\oldfigure\beginfloat}
\def\figure{\@ifnextchar[\figure@i \figure@ii}
\def\table@i[#1]{\oldtable[#1]\beginfloat}
\def\table@ii{\oldtable\beginfloat}
\def\table{\@ifnextchar[\table@i \table@ii}
\makeatother
\newcommand\beforefloat{\forcefontspace\makeparspace}
\newcommand\afterfloat{\vspace{.02em}}
\BeforeBeginEnvironment{figure}{\beforefloat}
\BeforeBeginEnvironment{table}{\beforefloat}
\AfterEndEnvironment{figure}{\afterfloat}
\AfterEndEnvironment{table}{\afterfloat}
\newlength{\parskipcopy}
\makeatletter
\def\@minipagerestore{\setlength{\intextsep}{0em}\setlength{\parskip}{\parskipcopy}\vphantom{ABCDEgjpqy}\vspace{-\baselineskip}\vspace{-\parskip}}
\AtEndEnvironment{minipage}{\forcefontspace}
\AfterEndEnvironment{minipage}{\makelinespace}
\let\oldfbox\fbox
\renewcommand{\fbox}[1]{\vspace{.05em}\setlength{\fboxsep}{0em}\oldfbox{#1}\vspace{.2em}\ifvmode{\makelinespace\ensurelinespace}\fi}
\renewcommand{\boxed}[1]{\oldfbox{\m@th$#1$}}	
\makeatother

\newtheoremstyle{lwq}
  {0em}	
  {0em}	
  {\normalfont}	
  {0em}	
  {\bfseries}	
  {.}	
  {.3em plus .2em}	
  {\thmname{#1}\thmnumber{ #2}\thmnote{ (#3)}}	
\newtheoremstyle{lwqprf}
  {0em}	
  {0em}	
  {\normalfont}	
  {0em}	
  {\itshape}	
  {.}	
  {.3em plus .2em}	
  {\thmname{#1}\thmnote{ (#3)}}	
\newlength{\thmspace}
\newlength{\prfspace}
\newcommand\thmbegin{\par\addvspace{\thmspace}\vspace{\parskip}\addpenalty{-200}}
\newcommand\prfbegin{\par\addvspace{\prfspace}\vspace{\parskip}}
\newcommand\thmend{\par\addvspace{\thmspace}\addpenalty{-100}}	
\newcommand\prfend{\par\addvspace{\prfspace}\addpenalty{-100}}	
\newenvironment{centerbox}
{\par\begin{centering}}
{\par\end{centering}}
\newcommand{\theoremname}{Theorem}
\theoremstyle{lwq}\newtheorem{thm}{\protect\theoremname}
\AtBeginEnvironment{thm}{\thmbegin} \AtEndEnvironment{thm}{\thmend}
\newcommand{\lemmaname}{Lemma}
\theoremstyle{lwq}\newtheorem{lem}[thm]{\protect\lemmaname}
\AtBeginEnvironment{lem}{\thmbegin} \AtEndEnvironment{lem}{\thmend}
\renewcommand{\proofname}{Proof}
\theoremstyle{lwqprf}\newtheorem{prf}{\protect\proofname}
\renewenvironment{proof}[1][]{\prfbegin\begin{prf}[#1]\pushQED{\qed}}{\popQED\end{prf}\prfend}
\renewcommand{\qed}{\hfill{}\hspace{2em minus 1em}\qedsymbol}
\renewcommand{\qedsymbol}{\scalebox{1.4}{$\diamond$}}
\newcommand{\remarkname}{Remark}
\theoremstyle{lwq}
\AtBeginEnvironment{rem}{\thmbegin} \AtEndEnvironment{rem}{\thmend}
\theoremstyle{lwq}\newtheorem*{rem*}{\protect\remarkname}
\AtBeginEnvironment{rem*}{\thmbegin} \AtEndEnvironment{rem*}{\thmend}
\usepackage[framemethod=TikZ]{mdframed}
\mdfdefinestyle{mdroundedbox}
{
  skipabove=0em,
  skipbelow=0em,
  roundcorner=.3em,
  innertopmargin=.3em,
  innerbottommargin=.3em,
  innerrightmargin=.3em,
  innerleftmargin=.3em,
  backgroundcolor=none,
  linecolor=gray,
}
\newenvironment{roundedbox}
{\forcefontspace\makeparspace\vspace{.1em}\begin{mdframed}[style=mdroundedbox]\vspace{.1em}\forcefontspace}
{\forcefontspace\end{mdframed}\unskip\vspace{-.22em}\forcefontspace}
\usepackage[framemethod=TikZ]{mdframed}
\mdfdefinestyle{mdroundedboxinfloat}
{
  skipabove=0em,
  skipbelow=0em,
  roundcorner=.3em,
  innertopmargin=.3em,
  innerbottommargin=.3em,
  innerrightmargin=.3em,
  innerleftmargin=.3em,
  backgroundcolor=none,
  linecolor=gray,
}
\newenvironment{roundedboxinfloat}
{\vspace{-.3em}\begin{mdframed}[style=mdroundedboxinfloat]\vspace{.1em}\forcefontspace}
{\forcefontspace\end{mdframed}\unskip\vspace{-.24em}}

\AtBeginDocument{

}

\makeatother

\begin{document}
\global\long\def\nn{\mathbb{N}}
\global\long\def\zz{\mathbb{Z}}
\global\long\def\qq{\mathbb{Q}}
\global\long\def\rr{\mathbb{R}}
\global\long\def\pp{\mathbb{P}}

\global\long\def\wi{\subseteq}
\global\long\def\co{\supseteq}
\global\long\def\nwi{\nsubseteq}
\global\long\def\nco{\nsupseteq}
\global\long\def\none{\varnothing}
\global\long\def\less{\smallsetminus}

\global\long\def\g{\mathbb{G}}
\global\long\def\h{\mathbb{H}}

\global\long\def\scale#1#2{\scalebox{#2}{\ensuremath{#1}}}
\global\long\def\tilt#1#2{\rotatebox{#2}{\ensuremath{#1}}}
\global\long\def\vsp#1{\vspace{#1em}}

\newcommand\br{\addpenalty{-1000}}

\def\boxheader#1{\texorpdfstring{#1}{[#1]}}
\def\thesect{\arabic{section}}
\def\thesubsect{\arabic{section}.\arabic{subsection}}
\def\thesubsubsect{\arabic{section}.\arabic{subsection}.\arabic{subsubsection}}
\def\thesection{\boxheader{\thesect}}
\def\thesubsection{\boxheader{\thesubsect}}
\def\thesubsubsection{\boxheader{\thesubsubsect}}
\titlelabel{\smash{\boxed{\textbf{\thetitle}}}\quad}

\newcommand\case{\crefname{enumi}{case}{}}

\noindent \begin{center}
\textbf{\LARGE{}Small Connected Planar Graph with 1-Cop-Move Number
4}
\par\end{center}{\LARGE \par}

\noindent \begin{center}
\begin{tabular}{>{\centering}p{0.3\textwidth}}
\textbf{\large{}Wei Quan Lim}{\large \par}

National University of Singapore\tabularnewline
\end{tabular}
\par\end{center}

\section*{Abstract}

This paper describes a $720$-vertex connected planar graph $\g$
such that $cop_{1}(\g)$, denoting the minimum number of cops needed
to catch the robber in the 1-cop-move game on $\g$, is at least $4$
and at most $7$. Furthermore, $\g$ has a connected subgraph $\h$
such that $cop_{1}(\h)$ is exactly $4$, meaning that $4$ cops are
barely sufficient to catch the robber in the 1-cop-move game on $\h$.
This is a significant improvement over the graph given by Gao and
Yang in 2017~\cite{gao2017onecopplanar}.

\section*{Acknowledgements}

I would like to thank Ziyuan Gao for introducing me to this interesting
problem in the first place, as well as for giving me very helpful
feedback on drafts of this paper. This research was partly supported
by Singapore MOE AcRF Tier 2 project MOE2018-T2-1-160.

\section{Introduction}

The abstract game of Cops and Robbers is a perfect-information 2-player
game on a graph $G$ with two other parameters $c$ and $k$. Player
$C$ (for Cops), has $c$ cops, and first places each of them at a
vertex of $G$. Player $R$ (for Robbers) then places the robber at
a (different) vertex of $G$. After that, $C$ and $R$ take turns
to make a move. On $R$'s turn, $R$ may move the robber by $1$ step,
namely from its current vertex along an edge to a neighbouring vertex.
On $C$'s turn, $C$ may move up to $k$ cops, each by $1$ step.
The classical variant where $k=c$ was introduced decades ago~\cite{aigner1984copnumberplanar},
whereas the variant where $k=1$ has been the subject of mathematical
study only in the past few years~\cite{offner2014copvariants,bal2015onecophypercube,yang2015onecoptime,bal2016onecopsurfacegraph}.
In general, this game is called the \textbf{$k$-cop-move game} with
$c$ cops on $G$. If eventually some cop moves to the same vertex
as the robber, then $C$ wins, otherwise $R$ wins. (As defined here,
neither cops nor robbers are not forced to move on each turn. For
some other variants see \cite{bonato2011copvariants,frieze2012copfastrobber,offner2014copvariants}.)

A natural question is, how many cops are needed to catch the robber
on a given graph? Specifically, the classical \textbf{cop number}
for $G$, denoted by $\boldsymbol{cop}(G)$, is the minimum $c$ such
that player $C$ (Cops) wins (i.e.~has a winning strategy for) the
$c$-cop-move game with $c$ cops on $G$. And the \textbf{$k$-cop-move
number} for $G$, denoted by $\boldsymbol{cop}_{k}(G)$, is the minimum
$c$ such that $C$ wins the $k$-cop-move game with $c$ cops on
$G$. The class of graphs with cop number $c$ has been characterized
for $c=1$ by Nowakowski and Winkler~\cite{nowakowski1983copnumberone}
and for general $c$ by Clarke and MacGillivray~\cite{clarke2012copnumberk}.
It is also natural to ask whether the cop number is bounded for the
class $\boldsymbol{\pp}$ of \textbf{\textit{finite connected planar
graphs}}, since the edge connections in a planar graphs are in some
sense local. Indeed, Aigner and Frommel showed that $cop(G)\le3$
for every graph $G$ in $\pp$. In contrast, much less is known about
the $1$-cop-move game for $\pp$~\cite{bonato2016onecopconjectures}.
Although Bal et~al.~\cite{bal2016onecopsurfacegraph} did show that
$cop_{1}(G)\in O\left(\sqrt{n}\right)$ for every graph $G$ in $\pp$
with $n$ vertices, it is conjectured that there is in fact a fixed
upper bound on $cop_{1}(G)$ for every graph $G$ in $\pp$, but this
remains unproven.

Recently, Gao and Yang constructed a graph $D$ with $cop_{1}(D)>3$~\cite{gao2017onecopplanar},
settling the question of whether there is even such a graph, which
was posed as an open problem by Sullivan et al.~\cite{sullivan2016onecopnumber3smallest}.
However, they were unable to prove their conjecture that $cop_{1}(D)=4$,
nor were they able to find a simpler construction. $D$ is constructed
by replacing each face of a dodecahedron with a number of nested pentagonal
layers, where the $k$-th layer from the centre has $20\cdot(k+1)$
vertices. In their paper, they used $49$ layers (in each face), resulting
in more than $280000$ vertices in $D$. It seems that, although the
number of layers can be reduced without essentially affecting their
solution, the resulting graph is likely to still have more than $10000$
vertices.

This paper provides an improved answer to that problem, namely a much
smaller graph $\g$ with merely $720$ vertices and a proof that $4\le cop_{1}(\g)\le7$,
as well as a connected subgraph $\h$ of $\g$ with $cop_{1}(\h)=4$.
It is hoped that the techniques used here, while somewhat ad-hoc,
will be helpful in figuring out the answer to the (still-open) question
of whether there is a graph with $1$-cop-move number $5$ or even
larger.

\section{The Construction}

\scalegraphics{.455}

\begin{minipage}[t]{0.49\columnwidth}%
To build the desired graph $\g$, we first start from the truncated
icosahedron $B$ (a.k.a.~the soccer ball graph) with $12$ pentagonal
faces and $20$ hexagonal faces, and retain its vertices but replace
its faces as depicted in the diagram on the right for one pentagonal
face and three of its neighbouring hexagonal faces.

The blue vertices (each with degree $6$ in $\g$) are the vertices
of $B$, and the black vertices are added vertices. Note that $\g$
has the same symmetries (i.e.~automorphism group) as $B$. There
are $15$ black vertices added to each pentagonal face, and $24$
black vertices added to each hexagonal faces, and $60$ blue vertices
in total, and so $\g$ has $15\cdot12+24\cdot20+60=720$ vertices
in all.%
\end{minipage}\hfill{}%
\begin{minipage}[t]{0.49\columnwidth}%
\begin{centerbox}
\begin{roundedbox}
\begin{centerbox}
\includegraphics[bb=14bp 14bp 490bp 343bp,clip]{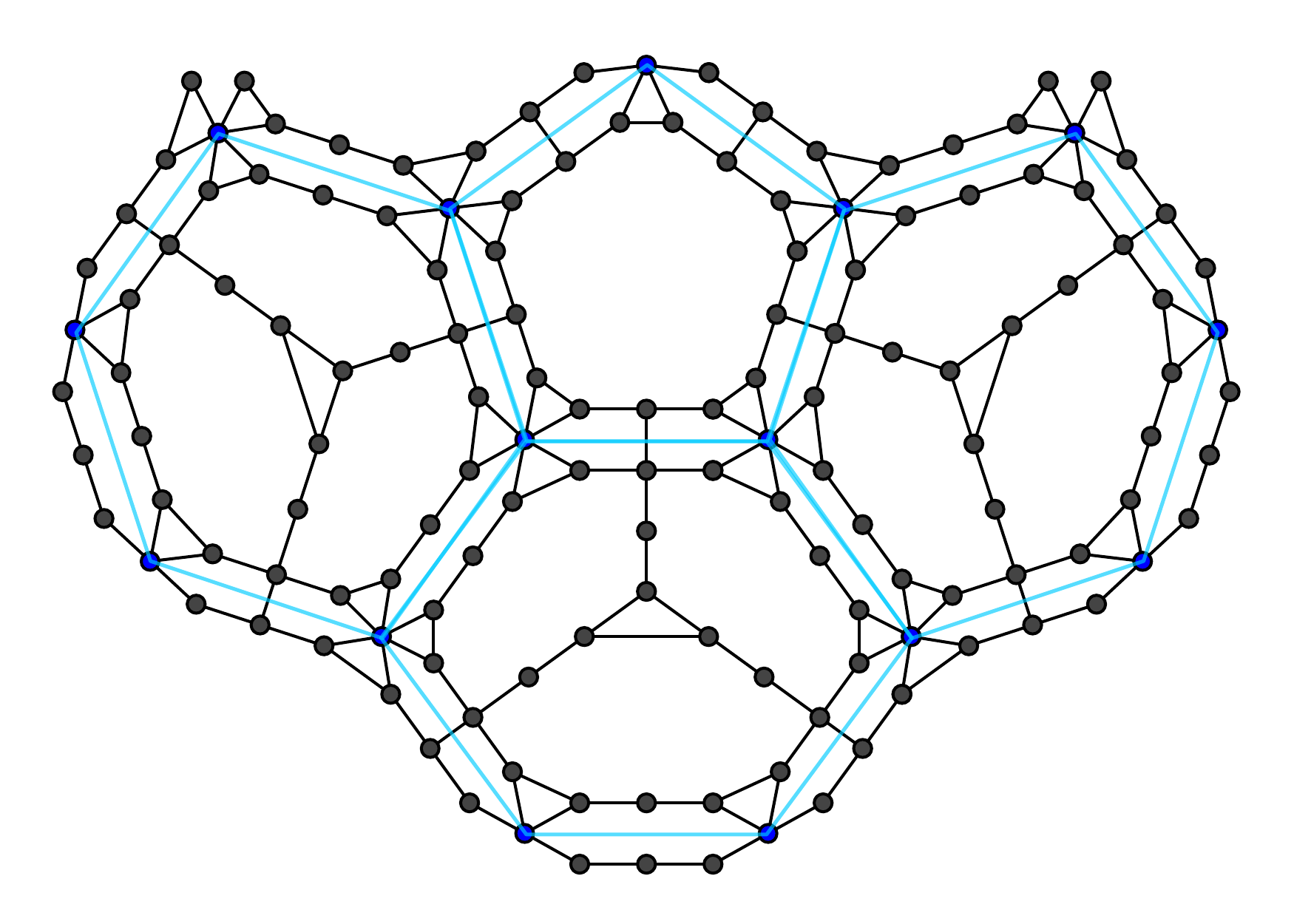}\end{centerbox}
\end{roundedbox}
\end{centerbox}
\end{minipage}

\section{The Robber Evades 3 Cops}

We shall now establish that $cop_{1}(\g)>3$, by explaining the winning
strategy for player $R$ (Robber). The general idea is for the robber
to stick to the \textbf{key vertices}, defined as the vertices of
$B$ (blue in the diagram), and move safely to another key vertex
whenever a cop gets too close, where \textbf{moving safely} to a vertex
$t$ means to move to $t$ in such a way that the cops cannot catch
the robber along the way and no cop is next to $t$ when the robber
reaches $t$. We shall use the following easy lemma throughout the
analysis.
\begin{lem}[Nearness Lemma]
\label{lem:nearness} On the robber's turn, if the robber is nearer
to a key vertex $t$ than any cop, then the robber can move safely
to $t$ by following any shortest path to $t$ without stopping.\end{lem}
\begin{proof}
Take any shortest path $P$ from the robber's starting vertex $v$
to $t$. For each vertex $w$ on $P$, just after the robber reaches
$w$, no cop can reach $w$ immediately after that, since its starting
vertex $u$ is further from $t$ than $v$ and so $d_{\g}(u,w)\ge d_{\g}(u,t)-d_{\g}(w,t)>d_{\g}(v,t)-d_{\g}(w,t)=d_{P}(v,w)$.
\end{proof}

Specifically, after the cops are placed, $R$ places the robber at
a key vertex that has no cop at or next to it (which is always possible
since each cop can be at or next to at most one key vertex), and then
over subsequent turns $R$ repeats the following indefinitely:

\begin{minipage}[t]{0.49\columnwidth}%
\vspace{\parskip}
\begin{enumerate}
\item \textbf{Stay phase}: Stay at the key vertex $v$ (i.e.~do not move
the robber) until a cop moves to a vertex $w$ adjacent to $v$. By
symmetry, there are essentially $3$ possible positions for $w$ relative
to $v$ as depicted in \ref{fig:cases}, and no cop is at any other
neighbour of $v$.
\item \textbf{Travel phase}: Let $A_{1},A_{2},A_{3}$ be the red-dotted
regions in \ref{fig:cases} on the lower-left, upper-left, and right
respectively. Each $A_{i}$ encloses vertices within $4$ steps from
some key vertex $v_{i}$ that is $4$ steps away from $v$, except
$v$ and some neighbours of $v$. There are two possible situations:

\begin{itemize}
\item There is exactly one cop in each of $A_{1},A_{2},A_{3}$.
\item There is no cop in some of $A_{1},A_{2},A_{3}$.
\end{itemize}
\item [] In either situation, it is possible to move safely to some key
vertex, as we shall show subsequently.\end{enumerate}
\end{minipage}\hfill{}%
\begin{minipage}[t]{0.49\columnwidth}%
\begin{figure}[H]
\begin{centerbox}
\begin{roundedboxinfloat}
\begin{centerbox}
\includegraphics[bb=14bp 14bp 490bp 343bp,clip]{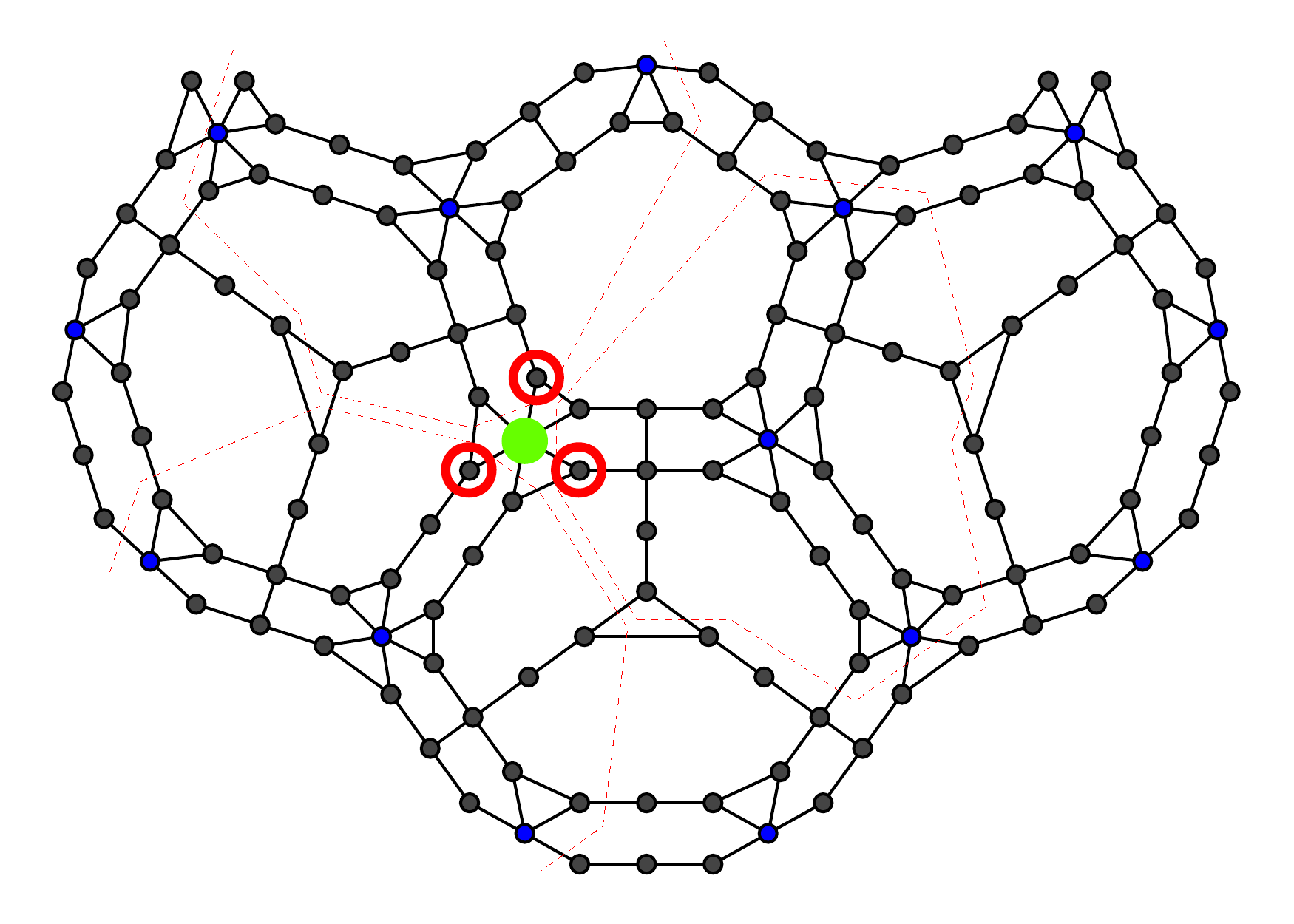}
\end{centerbox}
\caption{\label{fig:cases} At the start of the travel phase, the robber (solid
green circle) is at a key vertex $v$ and there is a cop at exactly
one neighbour $w$ of $v$, and there are $3$ possible positions
of $w$ (red-circled vertices).}
\end{roundedboxinfloat}
\end{centerbox}
\end{figure}
\end{minipage}

\subsection{One cop in each region}

We first deal with the situation where there is exactly one cop in
each of $A_{1},A_{2},A_{3}$. As noted earlier, there are essentially
$3$ cases for $w$ (and no cop is at any other neighbour of $v$):
\begin{enumerate}
\item $w$ is in $A_{2}$ and $1$ step away from $A_{3}$.
\item $w$ is in $A_{1}$ and $1$ step away from $A_{2}$.
\item $w$ is in $A_{3}$ and $1$ step away from $A_{1}$.
\end{enumerate}

\subsubsection{Case 1}

\label{sub:even1}

\begin{minipage}[t]{0.49\columnwidth}%
If the cop in $A_{1}$ is not at the X-marked vertex $x_{1}$ in the
diagram on the right, then the robber can use the green path to move
safely to one of the green-circled key vertices $t,v_{2}$ ($t$ is
the one on the left) or back to $v$.

More precisely, after the robber takes the first step along the green
path, if the cop at $w$ immediately starts moving along the red path,
then the robber can move back to $v$. Otherwise, the robber can continue
along the green path, and by the \nameref{lem:nearness} the cop that
was at $w$ must follow along the red path to guard $v_{2}$ (i.e.~prevent
the robber from moving safely to $v_{2}$), after which the robber
would be $4$ steps away from $t$ and the cops would all still be
at least $5$ steps away from $t$, since every cop in $A_{1}$ and
$A_{3}$ was initially at least $6$ steps away from $t$, so by the
\nameref{lem:nearness} the robber can move safely to $t$.%
\end{minipage}\hfill{}%
\begin{minipage}[t]{0.49\columnwidth}%
\begin{centerbox}
\begin{roundedbox}
\begin{centerbox}
\includegraphics[bb=14bp 14bp 490bp 343bp,clip]{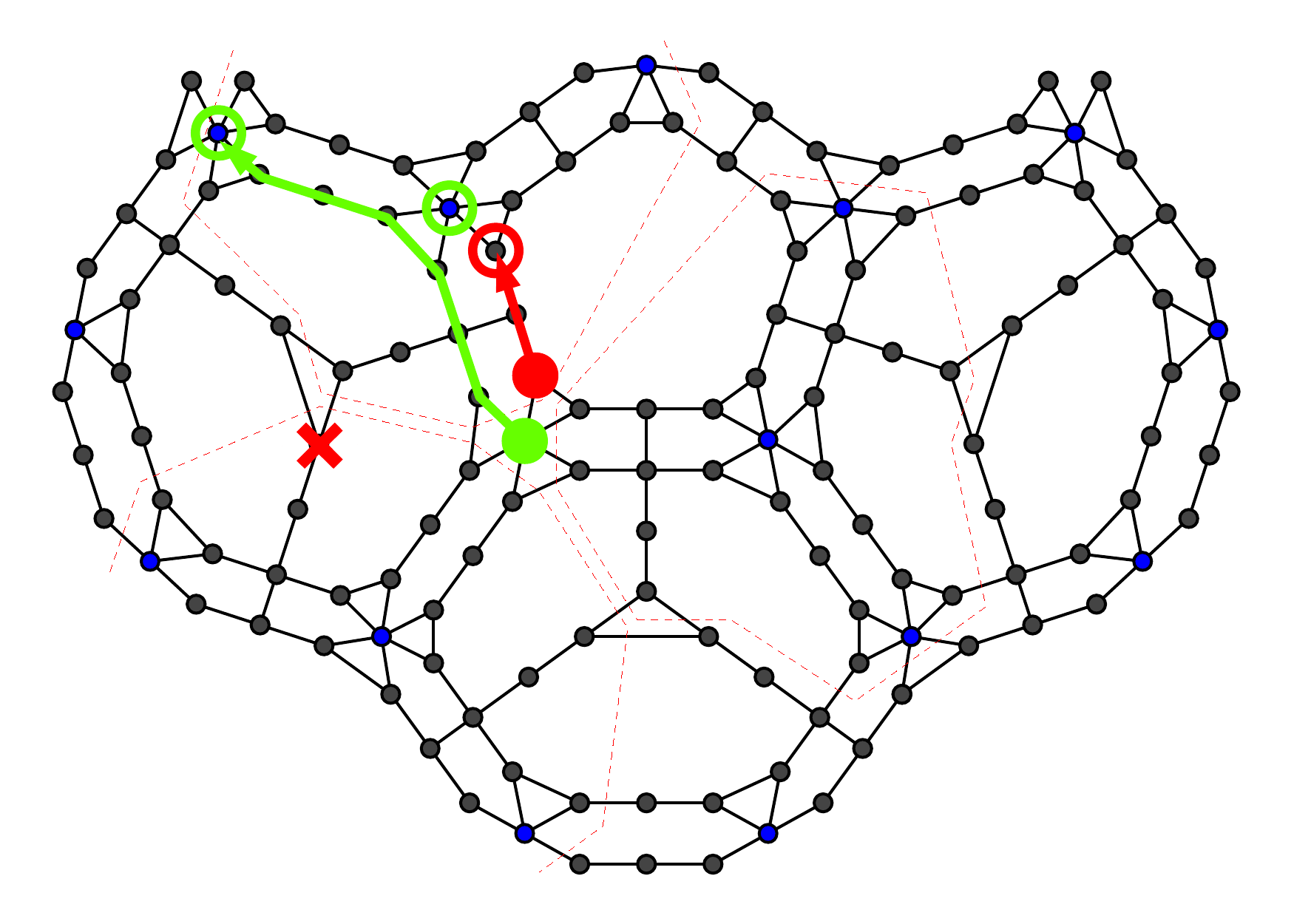}\end{centerbox}
\end{roundedbox}
\end{centerbox}
\end{minipage}

\hfill{} \dotfill{} \hfill{}

\begin{minipage}[t]{0.49\columnwidth}%
But if the cop in $A_{1}$ is at $x_{1}$, then the robber can instead
use the green path as shown on the right to move safely to one of
the green-circled key vertices $u,v_{1}$ ($u$ is the one on the
bottom).

More precisely, when the robber takes the first $3$ steps along the
green path, the cop that was at $x_{1}$ must follow along the red
path to guard $v_{1}$ by the \nameref{lem:nearness}, after which
the robber would be $4$ steps away from $u$ and the cops would all
still be at least $5$ steps away from $u$, so by the \nameref{lem:nearness}
the robber can move safely to $u$.%
\end{minipage}\hfill{}%
\begin{minipage}[t]{0.49\columnwidth}%
\begin{centerbox}
\begin{roundedbox}
\begin{centerbox}
\includegraphics[bb=14bp 14bp 490bp 343bp,clip]{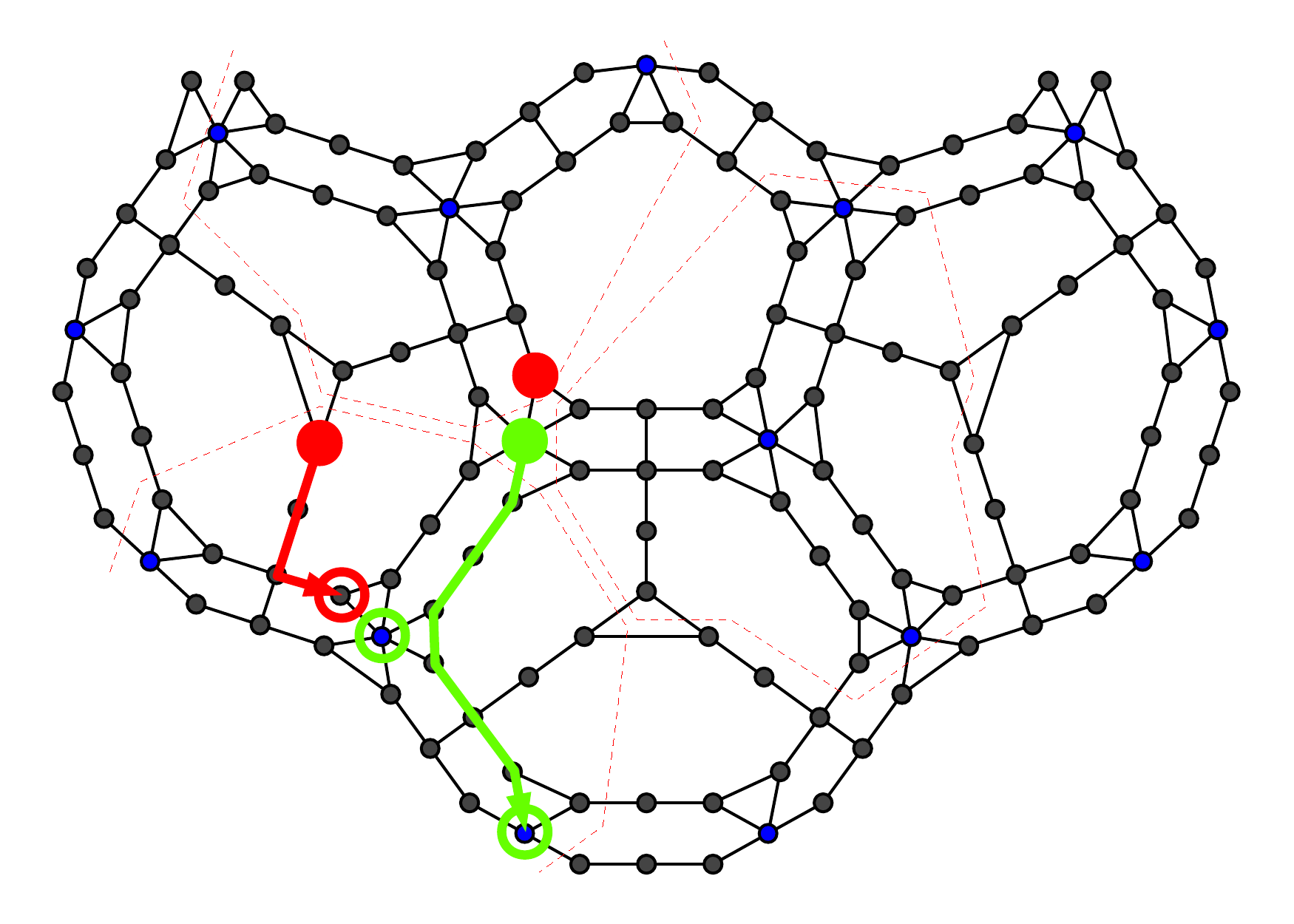}\end{centerbox}
\end{roundedbox}
\end{centerbox}
\end{minipage}

\clearpage{}

\subsubsection{Case 2}

\label{sub:even2}

\begin{minipage}[t]{0.49\columnwidth}%
If the cop in $A_{3}$ is not at the X-marked vertex $x_{3}$ in the
diagram on the right, then (exactly like in case 1a) the robber can
use the green path to move safely to one of the green-circled key
vertices $u,v_{1}$ ($u$ is the one on the bottom) or back to $v$.

More precisely, after the robber takes the first step along the green
path, if the cop at $w$ immediately starts moving along the red path,
then the robber can move back to $v$. Otherwise, the robber can continue
along the green path, and by the \nameref{lem:nearness} the cop that
was at $w$ must follow along the red path to guard $v_{1}$, after
which the robber would be $4$ steps away from $u$ and the cops would
all still be at least $5$ steps away from $u$, since every cop in
$A_{2}$ and $A_{3}$ was initially at least $6$ steps away from
$t$, so by the \nameref{lem:nearness} the robber can move safely
to $u$.%
\end{minipage}\hfill{}%
\begin{minipage}[t]{0.49\columnwidth}%
\begin{centerbox}
\begin{roundedbox}
\begin{centerbox}
\includegraphics[bb=14bp 14bp 490bp 343bp,clip]{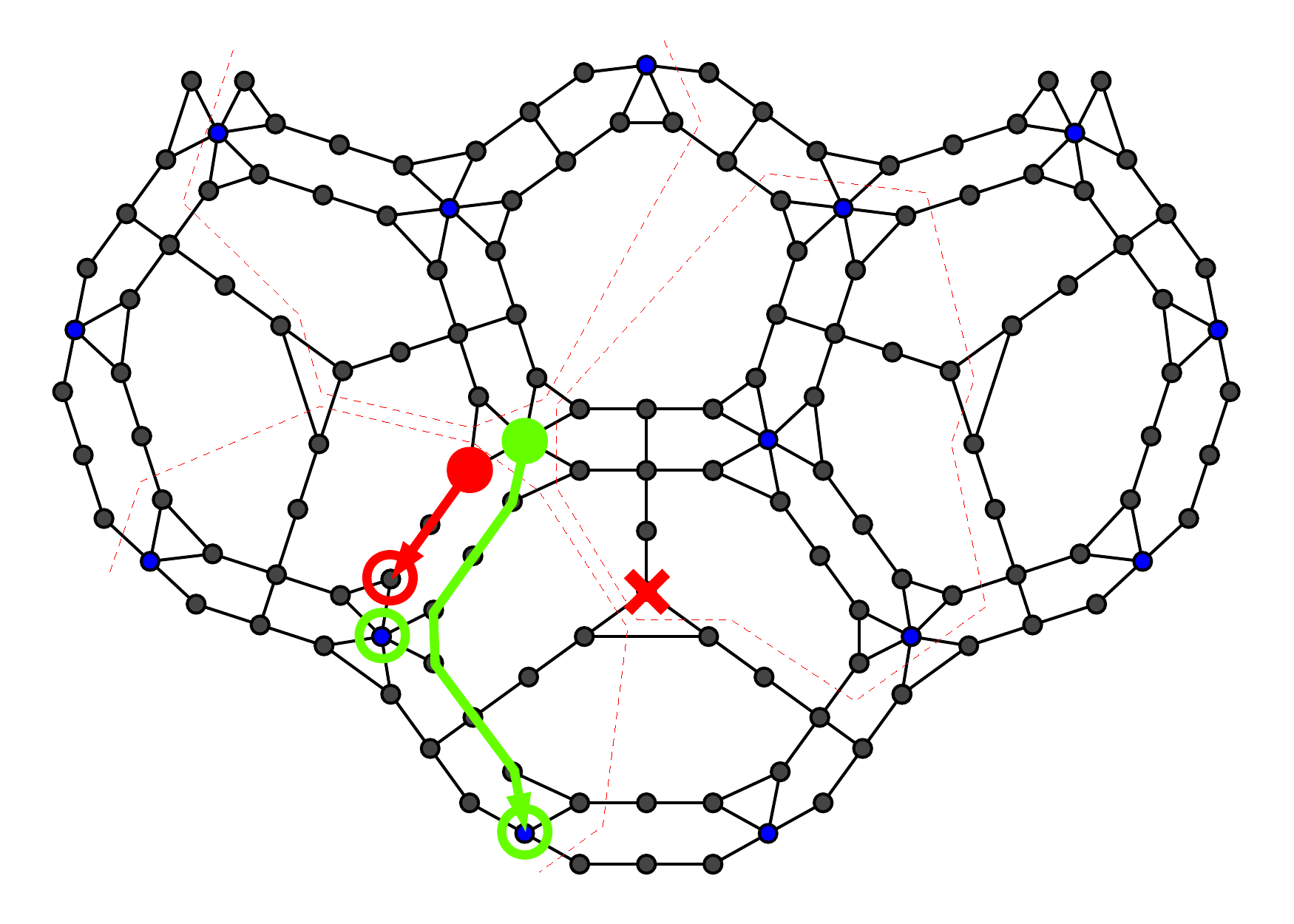}\end{centerbox}
\end{roundedbox}
\end{centerbox}
\end{minipage}

\hfill{} \dotfill{} \hfill{}

\begin{minipage}[t]{0.49\columnwidth}%
But if the cop in $A_{3}$ is at $x_{3}$, then the robber can instead
use the green path as shown on the right to move safely to the green-circled
vertex $m$, after which \textbf{\textit{either}} the robber can move
safely to $v_{3}$, \textbf{\textit{or}} the cop that was initially
at $x_{3}$ must next move to the red-circled vertex, and the other
cops are still in $A_{1}\cup A_{2}$ (and go to \nameref{sub:even4}).

More precisely, when the robber moves along the green path, by the
\nameref{lem:nearness} the cop that was at $x_{3}$ must follow along
the red path in order to guard $v_{3}$, during which no other cop
can move.%
\end{minipage}\hfill{}%
\begin{minipage}[t]{0.49\columnwidth}%
\begin{centerbox}
\begin{roundedbox}
\begin{centerbox}
\includegraphics[bb=14bp 14bp 490bp 343bp,clip]{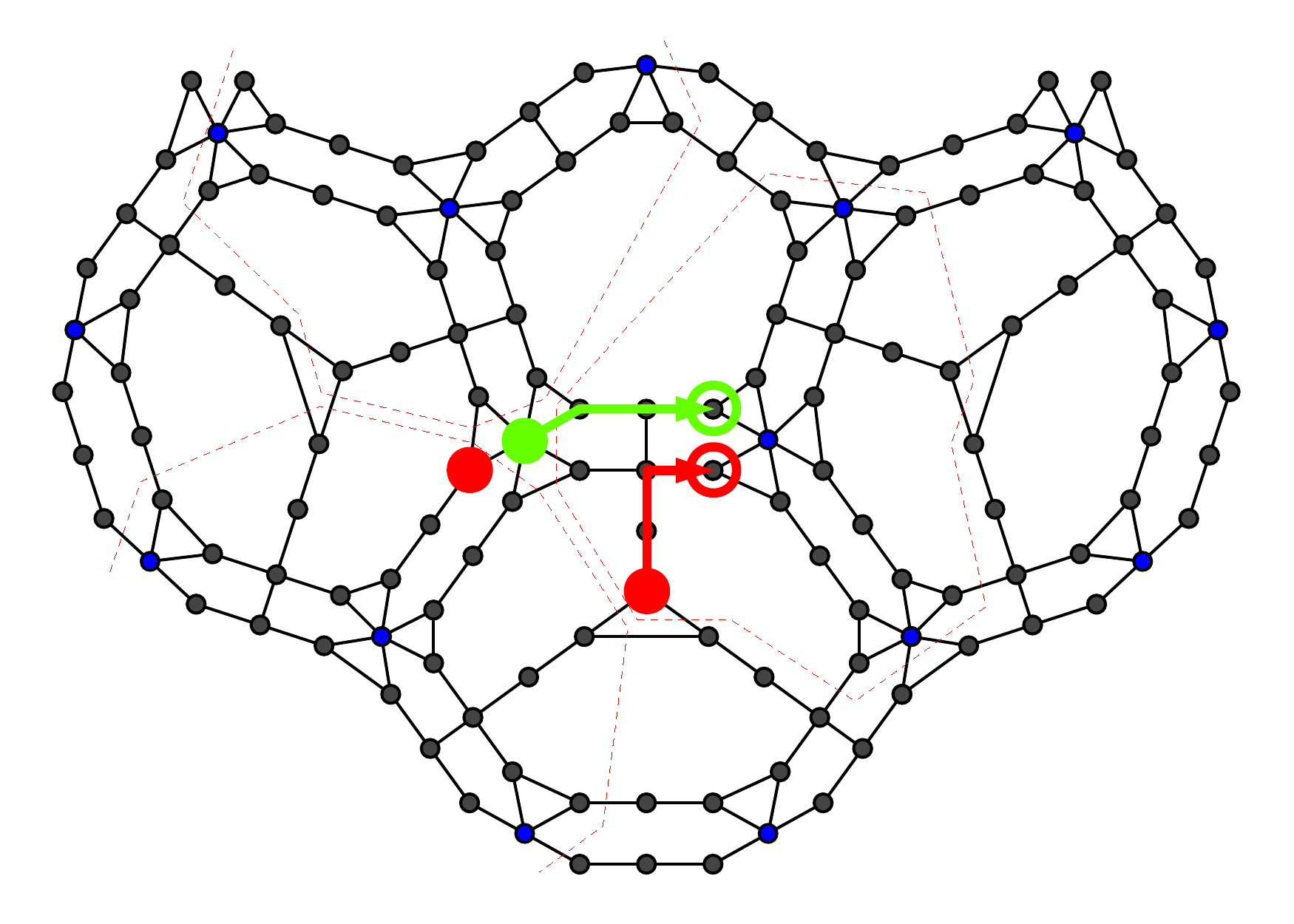}\end{centerbox}
\end{roundedbox}
\end{centerbox}
\end{minipage}

\subsubsection{Case 3}

\begin{minipage}[t]{0.49\columnwidth}%
If the cop in $A_{2}$ is not at the X-marked vertex $x_{2}$ in the
diagram on the right, then the robber can use the green path to move
safely to the green-circled vertex $m$ or back to $v$. And in the
former case, \textbf{\textit{either}} the robber can move safely to
$v_{3}$, \textbf{\textit{or}} the cop that was initially at $x_{2}$
must next move to the red-circled vertex, and the other cops are at
most $1$ step outside $A_{1}\cup A_{2}\less\{x_{2}\}$ (and go to
\nameref{sub:even4}).

More precisely, after the robber takes the first step along the green
path, if the cop at $w$ immediately starts moving along the red path,
then the robber can move back to $v$. Otherwise, the robber can continue
along the green path, and by the \nameref{lem:nearness} the cop that
was at $w$ must follow along the red path in order to guard $v_{3}$,
during which no other cop can move.%
\end{minipage}\hfill{}%
\begin{minipage}[t]{0.49\columnwidth}%
\begin{centerbox}
\begin{roundedbox}
\begin{centerbox}
\includegraphics[bb=14bp 14bp 490bp 343bp,clip]{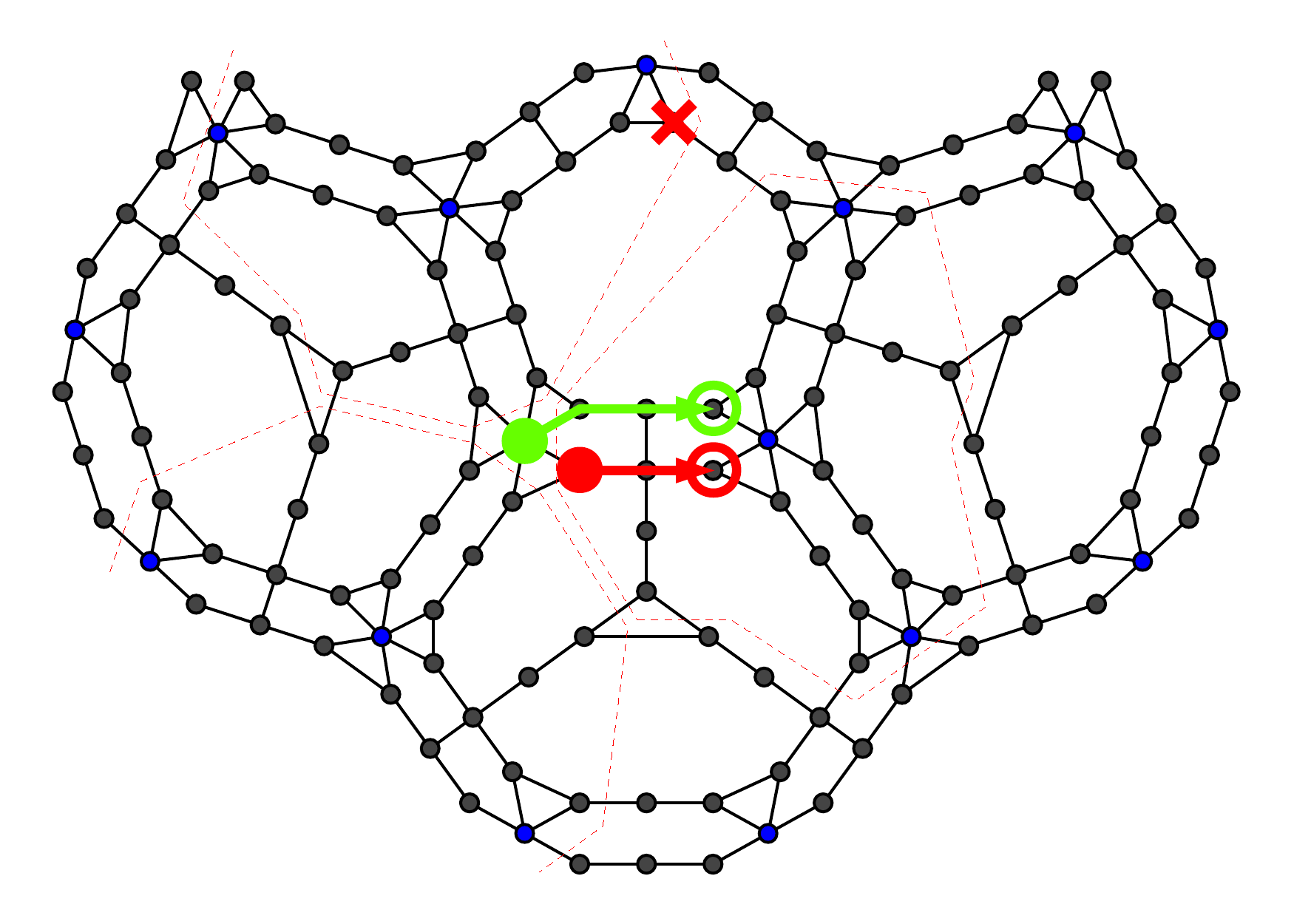}\end{centerbox}
\end{roundedbox}
\end{centerbox}
\end{minipage}

\hfill{} \dotfill{} \hfill{}

\begin{minipage}[t]{0.49\columnwidth}%
But if the cop in $A_{2}$ is at $x_{2}$, then (essentially like
in case 1b) the robber can instead use the green path as shown on
the right to move safely to one of the green-circled key vertices
$t,v_{2}$ ($t$ is the one on the left).

More precisely, when the robber takes the first $3$ steps along the
green path, the cop that was at $x_{2}$ must follow along the red
path to guard $v_{2}$ by the \nameref{lem:nearness}, after which
the robber would be $4$ steps away from $t$ and the cops would all
still be at least $5$ steps away from $t$, so by the \nameref{lem:nearness}
the robber can move safely to $t$.%
\end{minipage}\hfill{}%
\begin{minipage}[t]{0.49\columnwidth}%
\begin{centerbox}
\begin{roundedbox}
\begin{centerbox}
\includegraphics[bb=14bp 14bp 490bp 343bp,clip]{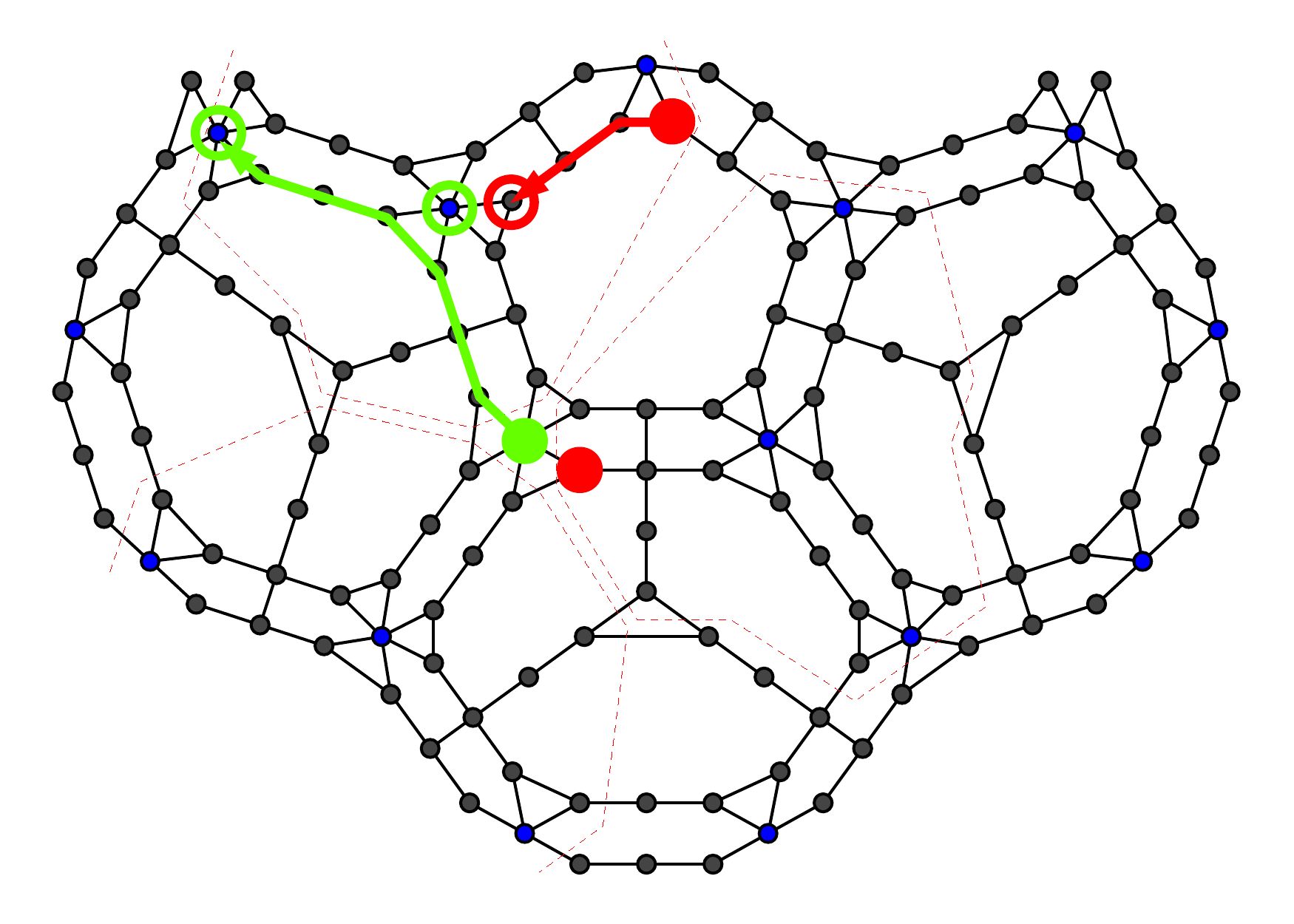}\end{centerbox}
\end{roundedbox}
\end{centerbox}
\end{minipage}

\subsubsection{Case 4}

\label{sub:even4}

\begin{minipage}[t]{0.49\columnwidth}%
The two unfinished cases above (i.e.~2b and 3a) can be handled in
the same way. The robber is now at vertex $m$ as shown on the right
(solid green circle) and is next to move. One cop is at a nearby vertex
$n$ (solid red circle), and the other two cops are each at a vertex
in one of the two red-dotted regions. From here, the robber can move
safely to one of the green-circled key vertices $v_{4},v_{5},v_{6}$
(named in clockwise order around the `hexagon' from the top-left).

To establish this, first observe that the robber can move safely along
the green path to the thin-green-circled vertex. After that, the cop
that was in $A_{2}$ must within the next move get to within $2$
steps from $v_{4}$ (i.e.~reach or pass a thin-red-circled vertex)
in order to guard it by the \nameref{lem:nearness}.

So if the robber cannot reach $v_{4}$ safely, the cop that was at
$n$ can move at most $1$ step so far, and hence the robber can continue
moving safely along the green path to the dotted-green-circled vertex.
At this point, the robber is only $5$ steps away from $v_{5}$, so
the cop from $A_{2}$ must within the next move get to within $5$
steps from $v_{5}$ (i.e.~reach a dotted-red-circled vertex) in order
to guard it by the \nameref{lem:nearness}, and hence must have taken
at least $5$ steps. But if the cop from $A_{2}$ does move in this
manner, then no other cop can have moved so far, and hence the robber
can safely move along the rest of the green path to $v_{6}$ by the
\nameref{lem:nearness}.%
\end{minipage}\hfill{}%
\begin{minipage}[t]{0.49\columnwidth}%
\begin{centerbox}
\begin{roundedbox}
\begin{centerbox}
\includegraphics[bb=14bp 14bp 490bp 343bp,clip]{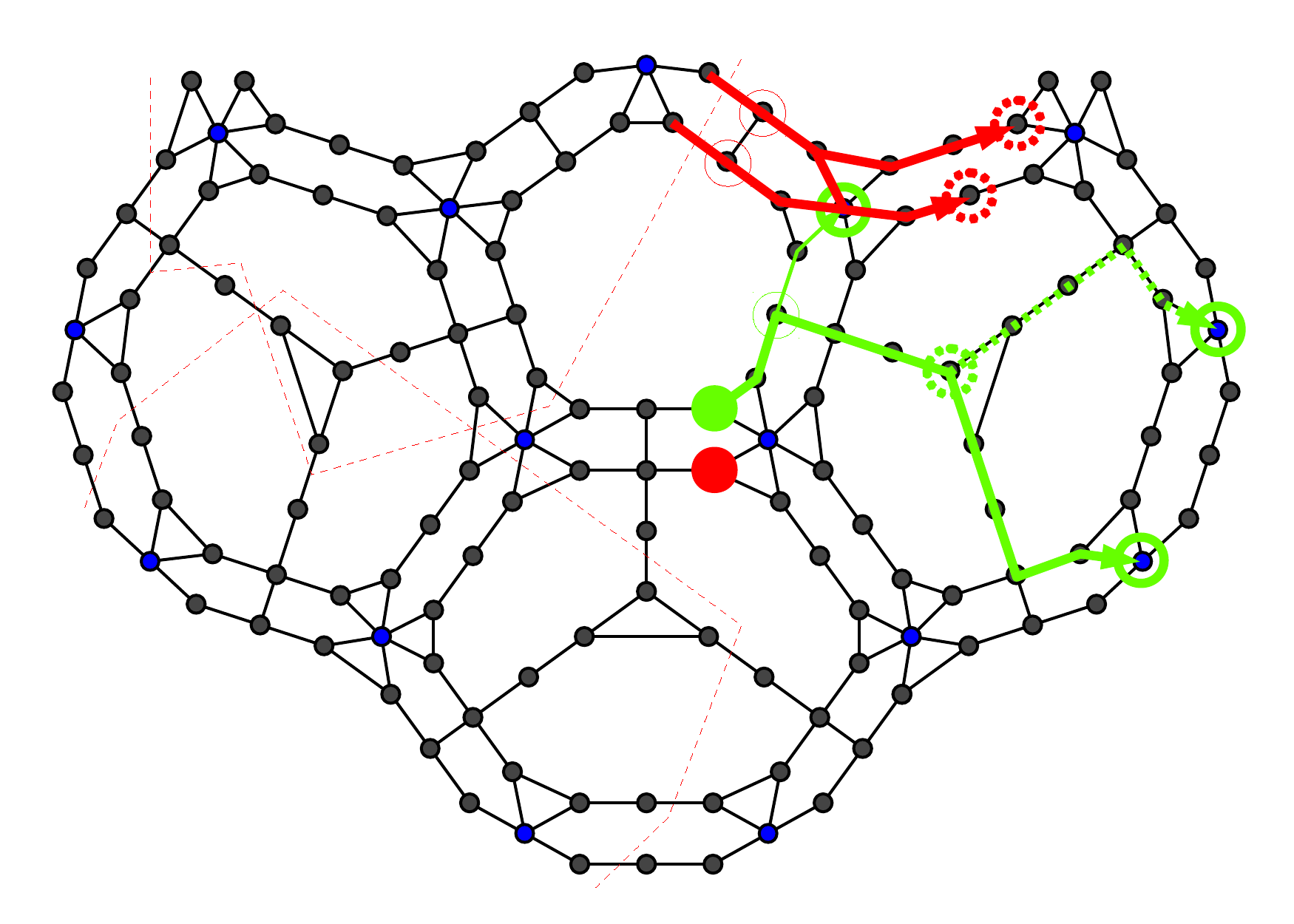}\end{centerbox}
\end{roundedbox}
\end{centerbox}
\end{minipage}

\subsection{No cop in some region}

We finally deal with the situation where there is no cop in $A_{i}$
for some $i\in\{1,2,3\}$. By symmetry, and since the cop at $w$
must be within $4$ steps from $v_{i}$ in order to guard $v_{i}$
by the \nameref{lem:nearness}, we only need to consider $3$ cases:
\begin{enumerate}
\item $A_{3}$ has no cop, and $w$ is in $A_{2}$ and just outside $A_{3}$.
\item $A_{3}$ has no cop, and $w$ is in $A_{1}$ and just outside $A_{3}$.
\item $A_{1}$ has no cop, and $w$ is in $A_{2}$ and just outside $A_{1}$.
\end{enumerate}

\subsubsection{Case 0}

\begin{minipage}[t]{0.49\columnwidth}%
Before analyzing those $3$ cases, we shall show how to handle a common
subcase. Here we assume case $1$, but it is essentially the same
in the other cases.

If no cop is exactly $2$ steps from $v$, then the robber can oscillate
between the green-circled vertices (see right) after moving to the
nearest one, as long as the cop that was at $w$ also oscillates between
the red-circled vertices. If the cops deviate from this, the robber
can thereafter move safely to either $v$ or $v_{3}$ by the \nameref{lem:nearness}.

Henceforth in all the $3$ subsequent cases we can assume that some
cop is exactly $2$ steps from $v$.%
\end{minipage}\hfill{}%
\begin{minipage}[t]{0.49\columnwidth}%
\begin{centerbox}
\begin{roundedbox}
\begin{centerbox}
\includegraphics[bb=14bp 14bp 490bp 343bp,clip]{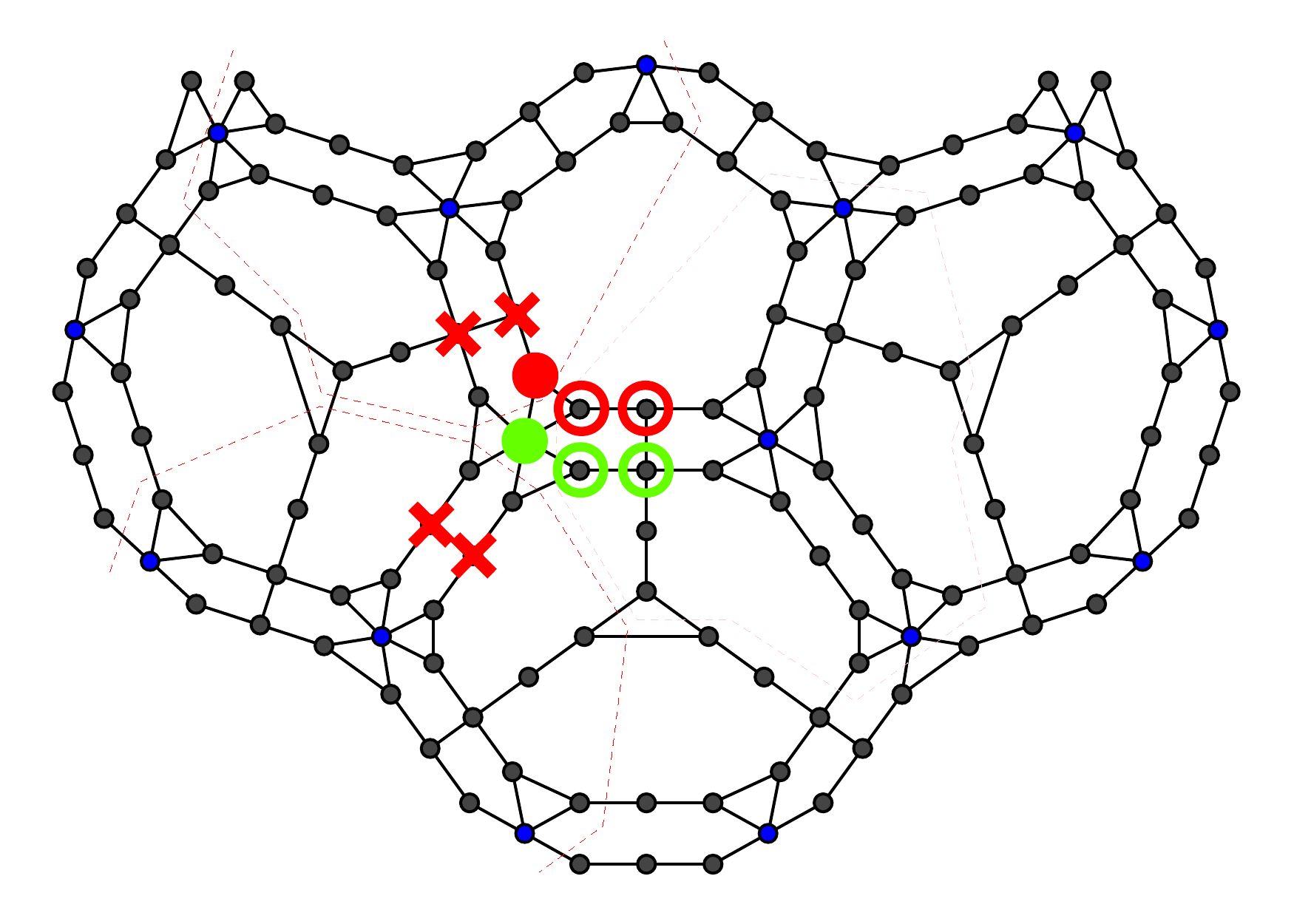}\end{centerbox}
\end{roundedbox}
\end{centerbox}
\end{minipage}

\subsubsection{Case 1}

\begin{minipage}[t]{0.49\columnwidth}%
By the \nameref{lem:nearness} there must be a cop in $A_{1}$, since
the cop at $w$ cannot guard $v_{1}$. Moreover, there must be a cop
in the yellow-dotted region $A_{4}$ (see right), which encloses vertices
outside $A_{3}$ that are within $4$ steps from the key vertex $v_{7}$
at the end of the green path, otherwise the robber can use the green
path to move safely to either $v_{3}$ or $v_{7}$ by the \nameref{lem:nearness},
since when the robber takes the first $3$ steps along the green path,
the cop that was at $w$ must follow along the red path to guard $v_{3}$,
after which the robber is only $4$ steps away from $v_{7}$.

Henceforth for the rest of this case we can assume that there is exactly
one cop in each of $A_{1},A_{2},A_{4}$.%
\end{minipage}\hfill{}%
\begin{minipage}[t]{0.49\columnwidth}%
\begin{centerbox}
\begin{roundedbox}
\begin{centerbox}
\includegraphics[bb=14bp 14bp 490bp 343bp,clip]{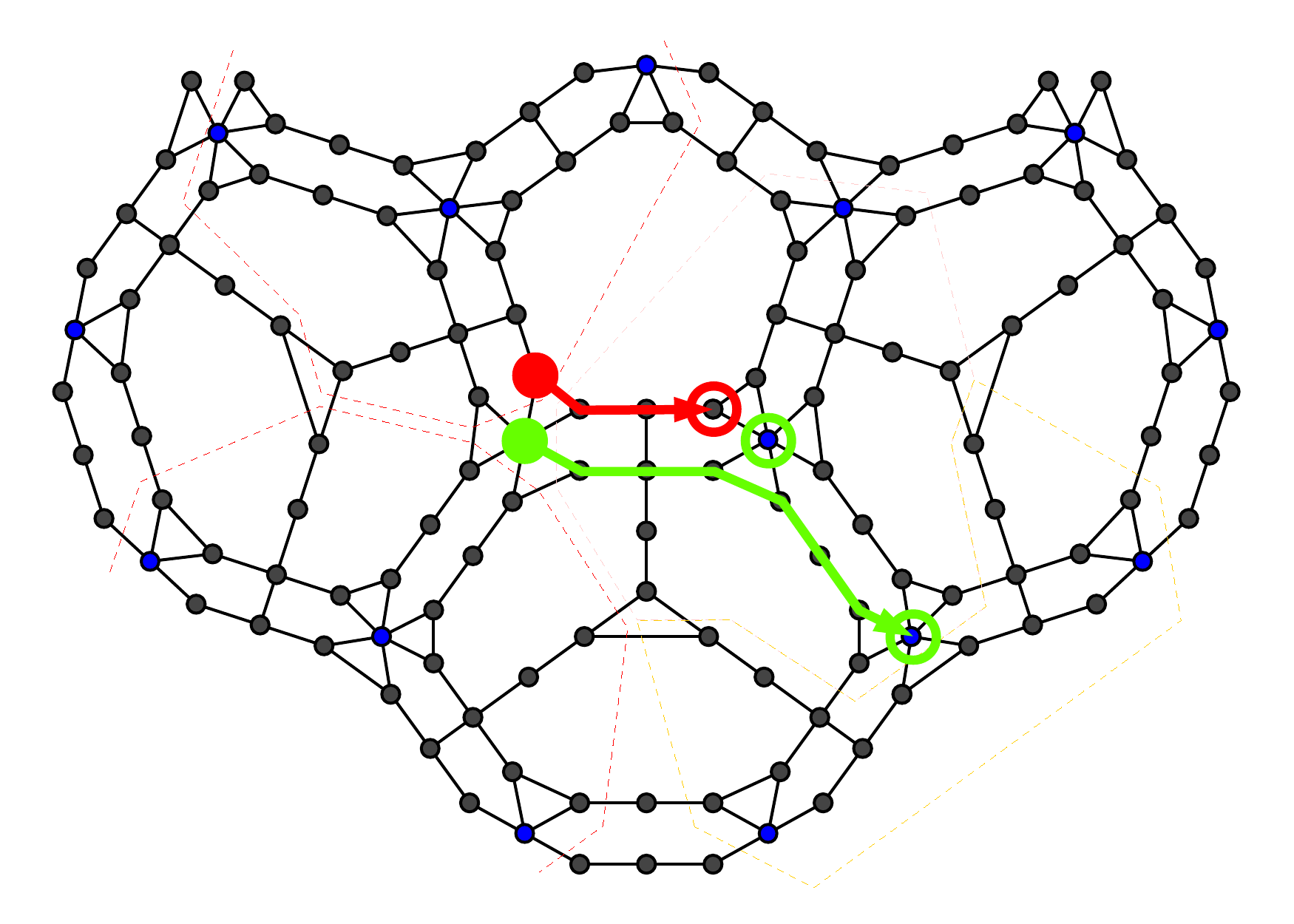}\end{centerbox}
\end{roundedbox}
\end{centerbox}
\end{minipage}

\hfill{} \dotfill{} \hfill{}

\begin{minipage}[t]{0.49\columnwidth}%
Since the cop in $A_{1}$ is exactly $2$ steps from $v$, and in
particular not at the X-marked vertex $x_{1}$ (see right), the robber
can use the green path to move safely to one of the key vertices $t,v_{2}$
or back to $v$, exactly like in \ref{sub:even1} (though the third
cop is in $A_{4}$ rather than $A_{3}$).%
\end{minipage}\hfill{}%
\begin{minipage}[t]{0.49\columnwidth}%
\begin{centerbox}
\begin{roundedbox}
\begin{centerbox}
\includegraphics[bb=14bp 14bp 490bp 343bp,clip]{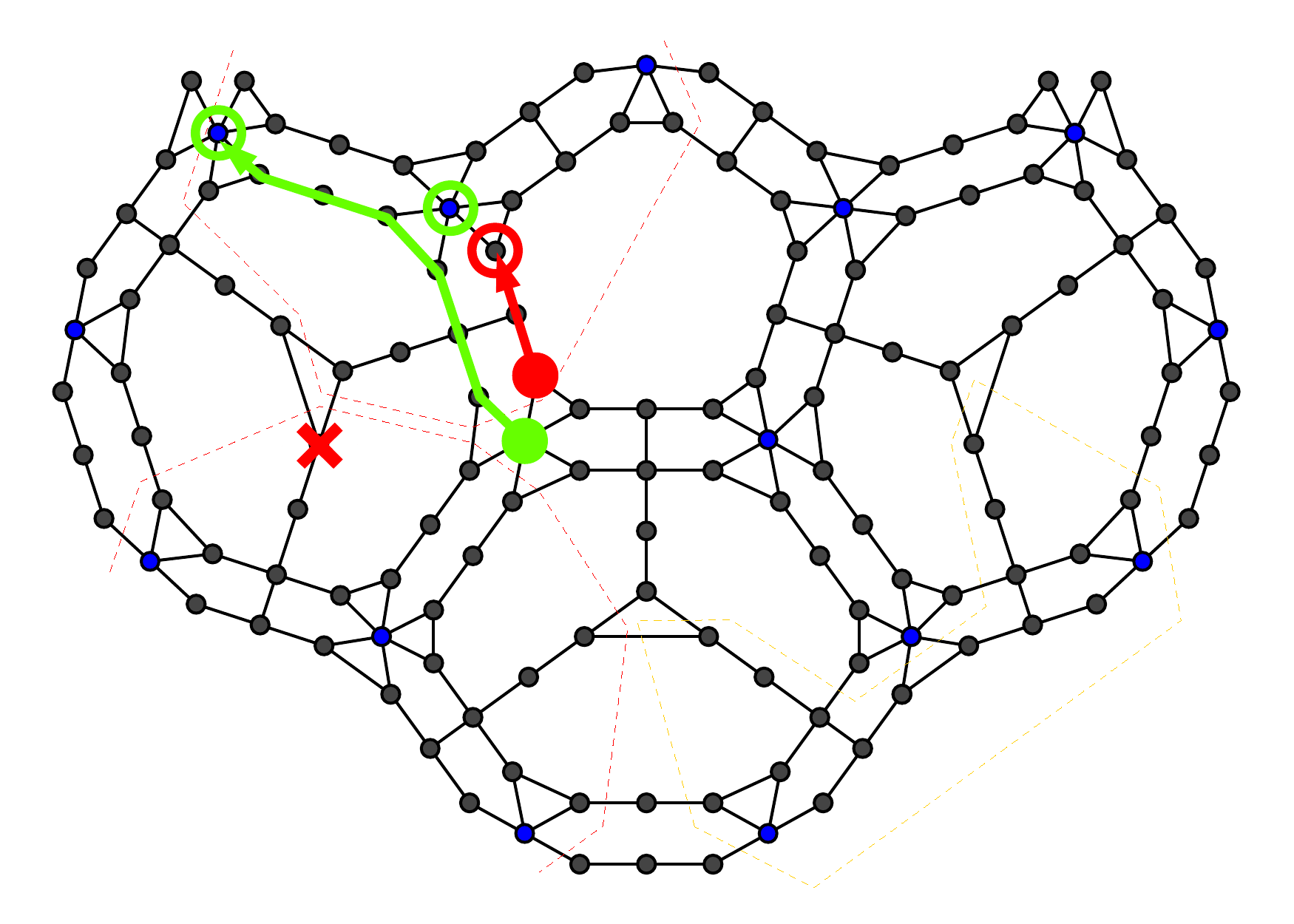}\end{centerbox}
\end{roundedbox}
\end{centerbox}
\end{minipage}

\subsubsection{Case 2}

\begin{minipage}[t]{0.49\columnwidth}%
By the \nameref{lem:nearness} there must be a cop in $A_{2}$, since
the cop at $w$ cannot guard $v_{2}$. Moreover, there must be a cop
in the yellow-dotted region $A_{5}$ (see right), which encloses vertices
outside $A_{3}$ that are within $4$ steps from the key vertex $v_{4}$
at the end of the green path, otherwise the robber can use the green
path to move safely to either $v_{3}$ or $v_{4}$ by the \nameref{lem:nearness},
since when the robber takes the first $3$ steps along the green path,
the cop that was at $w$ must follow along the red path to guard $v_{3}$,
after which the robber is only $4$ steps away from $v_{4}$.

Henceforth for the rest of this case we can assume that there is a
cop in each of $A_{1},A_{2},A_{5}$.%
\end{minipage}\hfill{}%
\begin{minipage}[t]{0.49\columnwidth}%
\begin{centerbox}
\begin{roundedbox}
\begin{centerbox}
\includegraphics[bb=14bp 14bp 490bp 343bp,clip]{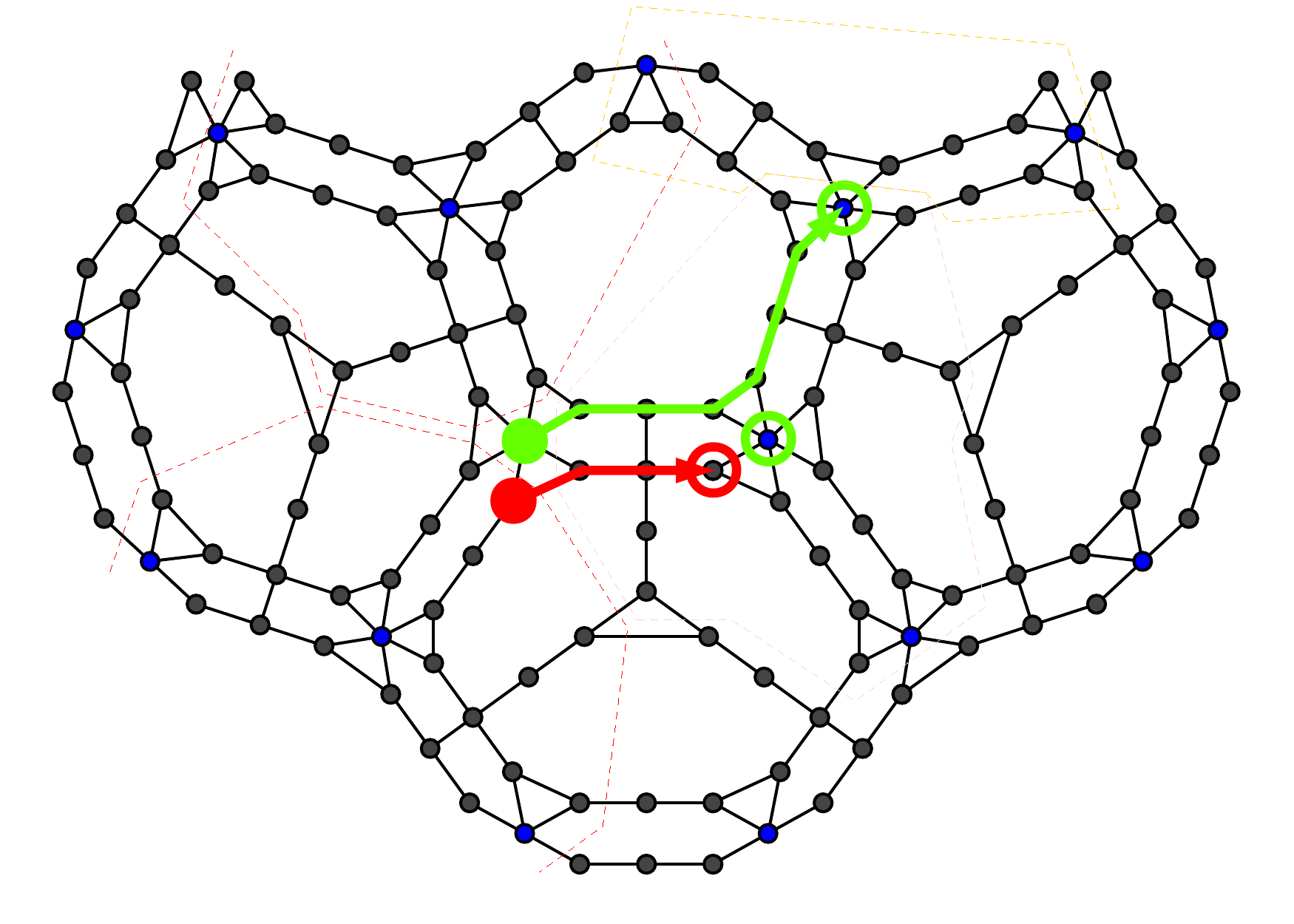}\end{centerbox}
\end{roundedbox}
\end{centerbox}
\end{minipage}

\hfill{} \dotfill{} \hfill{}

\begin{minipage}[t]{0.49\columnwidth}%
If the cop exactly $2$ steps from $v$ is in $A_{2}$, then there
is exactly one cop in each of $A_{1},A_{2},A_{5}$, and in particular
there is no cop at the X-marked vertex (see right), and so the robber
can use the green path to move safely to one of the green-circled
key vertices $u',v_{1}$ ($u'$ is the one on the left) or back to
$v$, exactly like in \ref{sub:even2} (though the third cop is in
$A_{5}$ rather than $A_{3}$).%
\end{minipage}\hfill{}%
\begin{minipage}[t]{0.49\columnwidth}%
\begin{centerbox}
\begin{roundedbox}
\begin{centerbox}
\includegraphics[bb=14bp 14bp 490bp 343bp,clip]{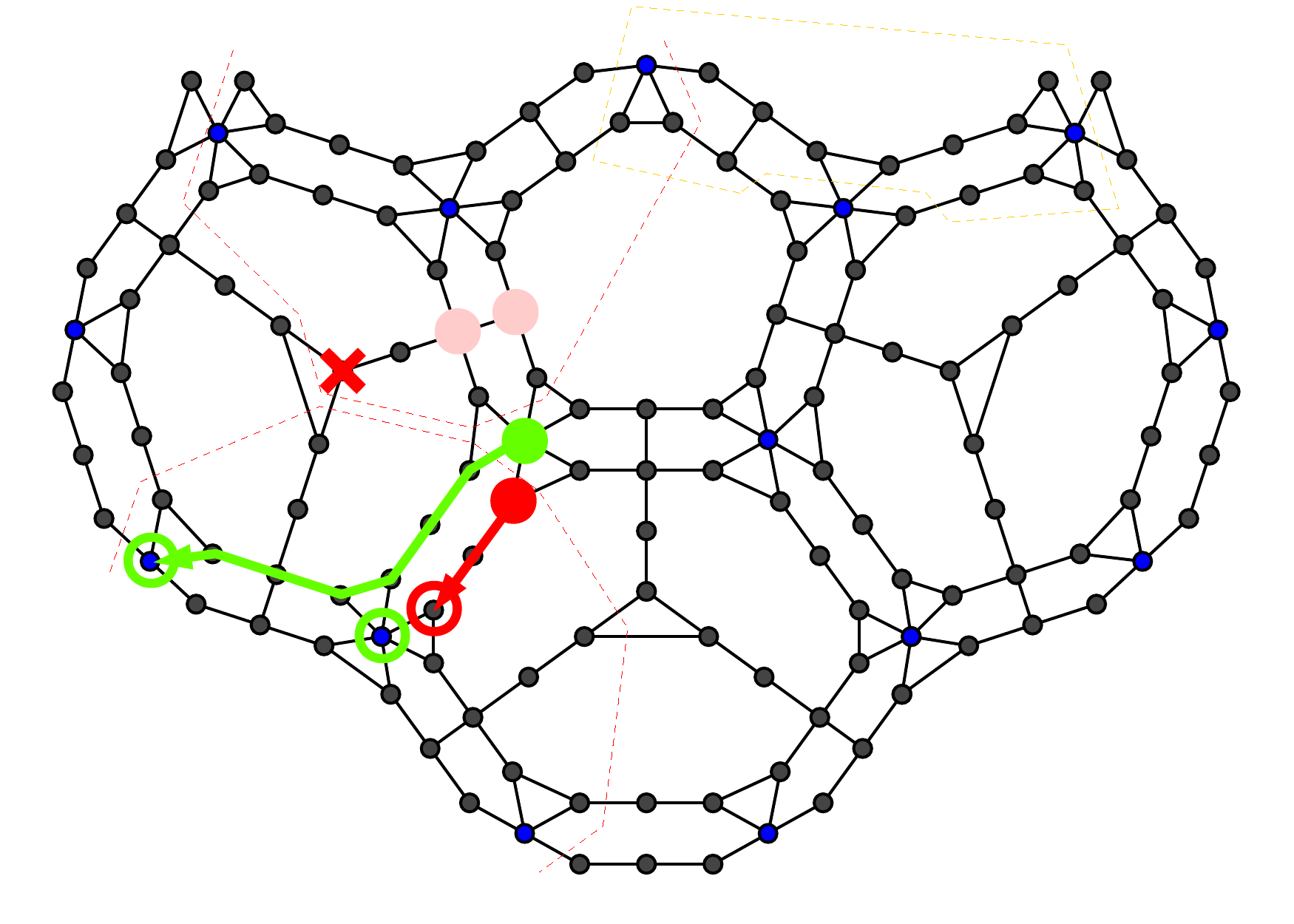}\end{centerbox}
\end{roundedbox}
\end{centerbox}
\end{minipage}

\hfill{} \dotfill{} \hfill{}

\begin{minipage}[t]{0.49\columnwidth}%
But if the cop exactly $2$ steps from $v$ is in $A_{1}$, then there
are already $2$ cops in $A_{1}$, and so the third cop must be in
$A_{2}\cap A_{5}$. Thus the robber can use the green path (see right)
to move safely to the green-circled vertex $m$, and the cop at $w$
must follow along the red path to the red-circled vertex $n$ in order
to guard $v_{3}$ by the \nameref{lem:nearness}, during which the
other cops cannot move and hence remain within $A_{1}\cup A_{2}$.

Thus after moving along the green path to $m$, if the robber cannot
reach $v_{3}$ safely in the next move, then on that move it must
be that one cop is at $n$ and the other two cops are in $A_{1}\cup A_{2}$,
and this situation is covered by \ref{sub:even4}.%
\end{minipage}\hfill{}%
\begin{minipage}[t]{0.49\columnwidth}%
\begin{centerbox}
\begin{roundedbox}
\begin{centerbox}
\includegraphics[bb=14bp 14bp 490bp 343bp,clip]{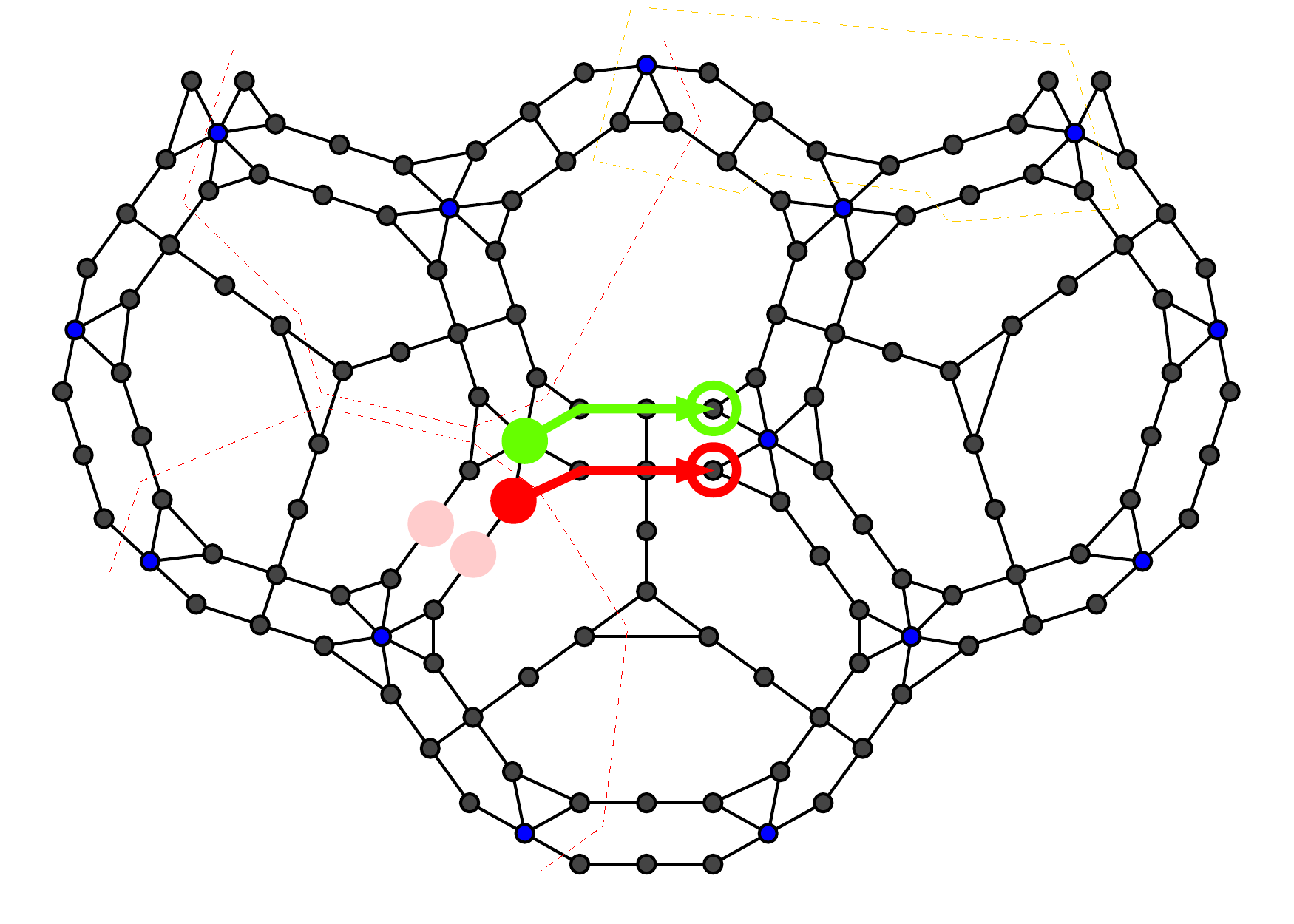}\end{centerbox}
\end{roundedbox}
\end{centerbox}
\end{minipage}

\subsubsection{Case 3}

\begin{minipage}[t]{0.49\columnwidth}%
By the \nameref{lem:nearness} there must be a cop in $A_{3}$, since
the cop at $w$ cannot guard $v_{3}$. Moreover, there must be a cop
in the yellow-dotted region $A_{6}$ (partly shown on the right) that
encloses vertices outside $A_{1}$ that are within $4$ steps from
the key vertex $u$ at the end of the green path, otherwise the robber
can use the green path to move safely to either $v_{1}$ or $u$ by
the \nameref{lem:nearness}, since when the robber takes the first
$3$ steps along the green path, the cop that was at $w$ must follow
along the red path to guard $v_{1}$, after which the robber is only
$4$ steps away from $u$.

Since $A_{2},A_{3},A_{6}$ are disjoint, we can for the rest of this
case assume that there is exactly one cop in each of $A_{2},A_{3},A_{6}$.%
\end{minipage}\hfill{}%
\begin{minipage}[t]{0.49\columnwidth}%
\begin{centerbox}
\begin{roundedbox}
\begin{centerbox}
\includegraphics[bb=14bp 14bp 490bp 343bp,clip]{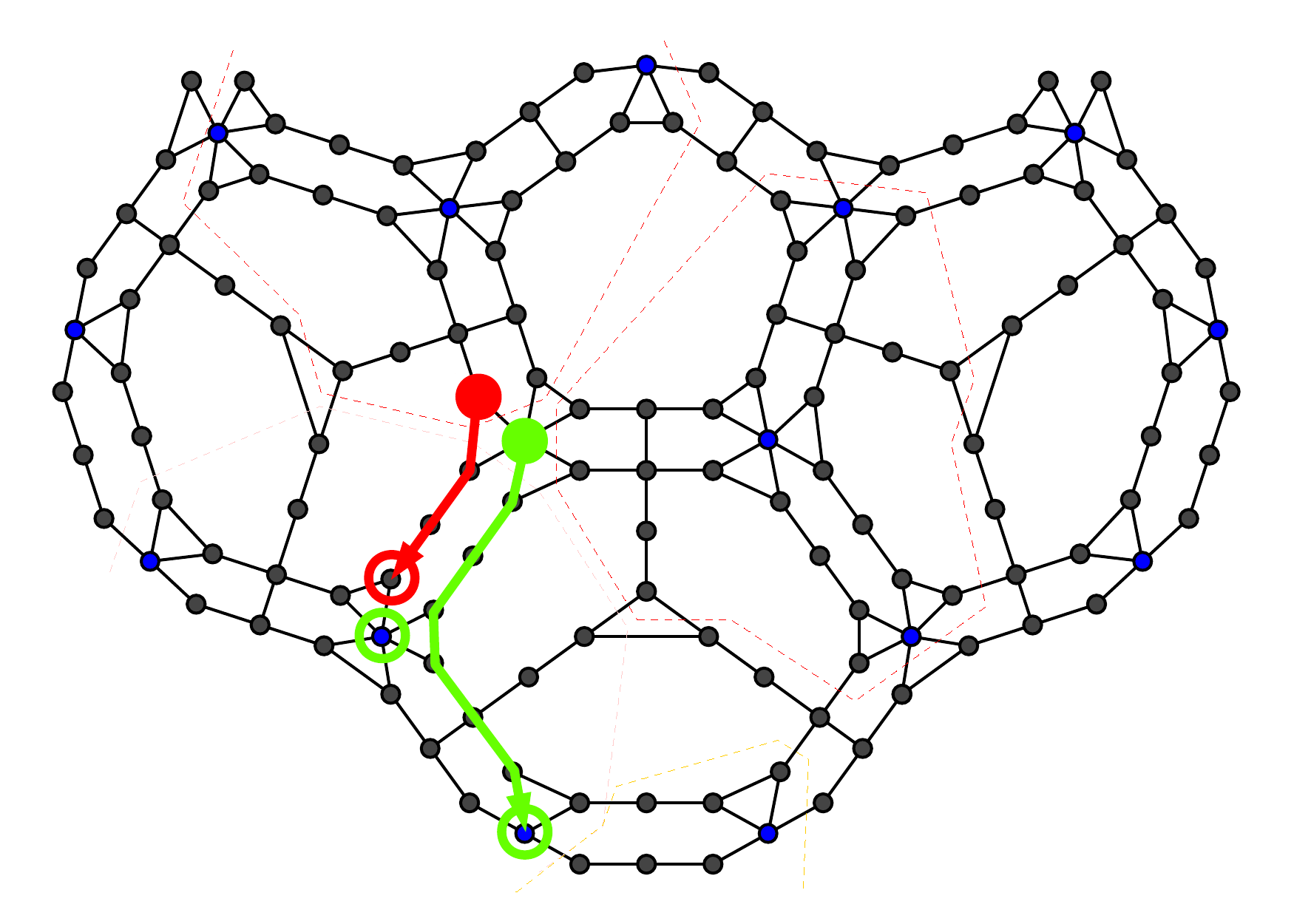}\end{centerbox}
\end{roundedbox}
\end{centerbox}
\end{minipage}

\hfill{} \dotfill{} \hfill{}

\begin{minipage}[t]{0.49\columnwidth}%
Since the cop in $A_{3}$ is exactly $2$ steps from $v$, and in
particular is at least $8$ steps away from the green-circled key
vertex $t'$ (see right), the robber can use the green path to move
safely to either $t'$ or back to $v$.

More precisely, after the robber takes $1$ step along the green path,
if the cop that was at $w$ moves away from $v$ then the robber can
move safely back to $v$, otherwise the robber can continue moving
safely along the green path since it is already only $6$ steps away
from $t'$.%
\end{minipage}\hfill{}%
\begin{minipage}[t]{0.49\columnwidth}%
\begin{centerbox}
\begin{roundedbox}
\begin{centerbox}
\includegraphics[bb=14bp 14bp 490bp 343bp,clip]{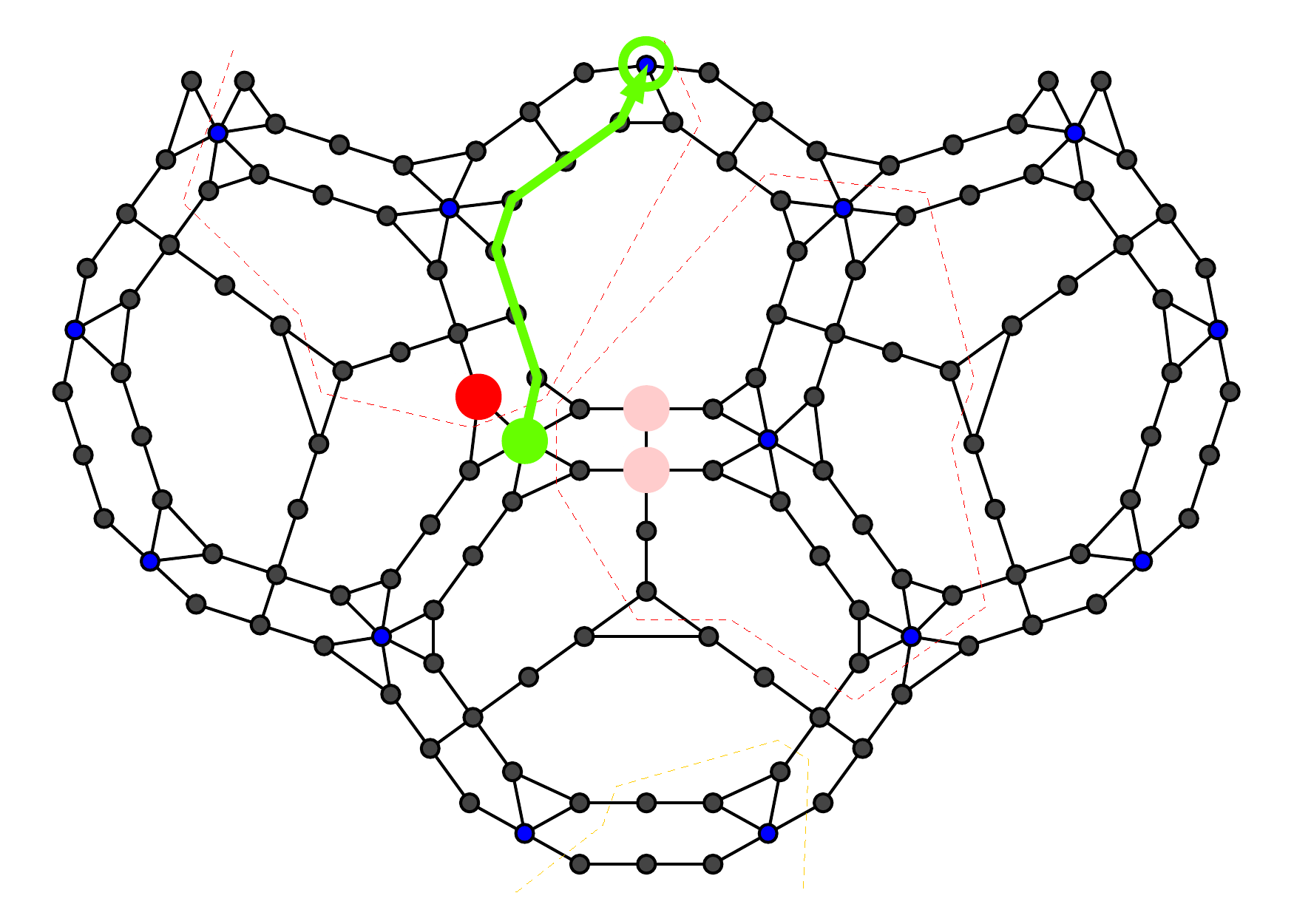}\end{centerbox}
\end{roundedbox}
\end{centerbox}
\end{minipage}

\clearpage{}

\section{7 Cops Catch the Robber}

\scalegraphics{.7}

Even though it seems like the robber can barely manage to escape from
$3$ cops using the strategy given in the previous section, it is
partly because in that strategy the robber waits until a cop is right
next to it, and there does not seem to be a concise strategy for $4$
cops to catch the robber, if there is one at all.

Nevertheless, it is not too hard to give a winning strategy for $7$
cops, establishing that $cop_{1}(\g)\le7$, which we shall do in this
section. The intuitive idea behind this strategy is to use some cops
to `guard' some vertices so as to restrict the robber to certain
possible regions. At the start we move the cops into an initial `guarding'
configuration, and thereafter in each phase we keep the robber `confined'
to a region using some cops while moving the other cops to new `guarding'
positions to `divide' that region, so that the robber would now
be `confined' to a smaller region.

\subsection{Hexagon Guarding}

\label{sub:hex-guard}

We begin with a lemma concerning how one cop can be used to guard
a hexagon (shown below in blue), namely to guard the three `sides'
of a hexagonal face of the truncated icosahedron $B$ that are adjacent
to the neighbouring pentagonal faces, in the sense of preventing the
robber from `crossing over'. The rough idea is that the cop will
try to stay in the \textbf{central vertices} of the hexagon, namely
at one of the three vertices of the triangle in the centre of the
hexagon, and move towards one `side' only when the robber gets close
to that side.
\begin{centerbox}
\begin{roundedbox}
\begin{centerbox}
\includegraphics[bb=14bp 14bp 490bp 343bp,clip]{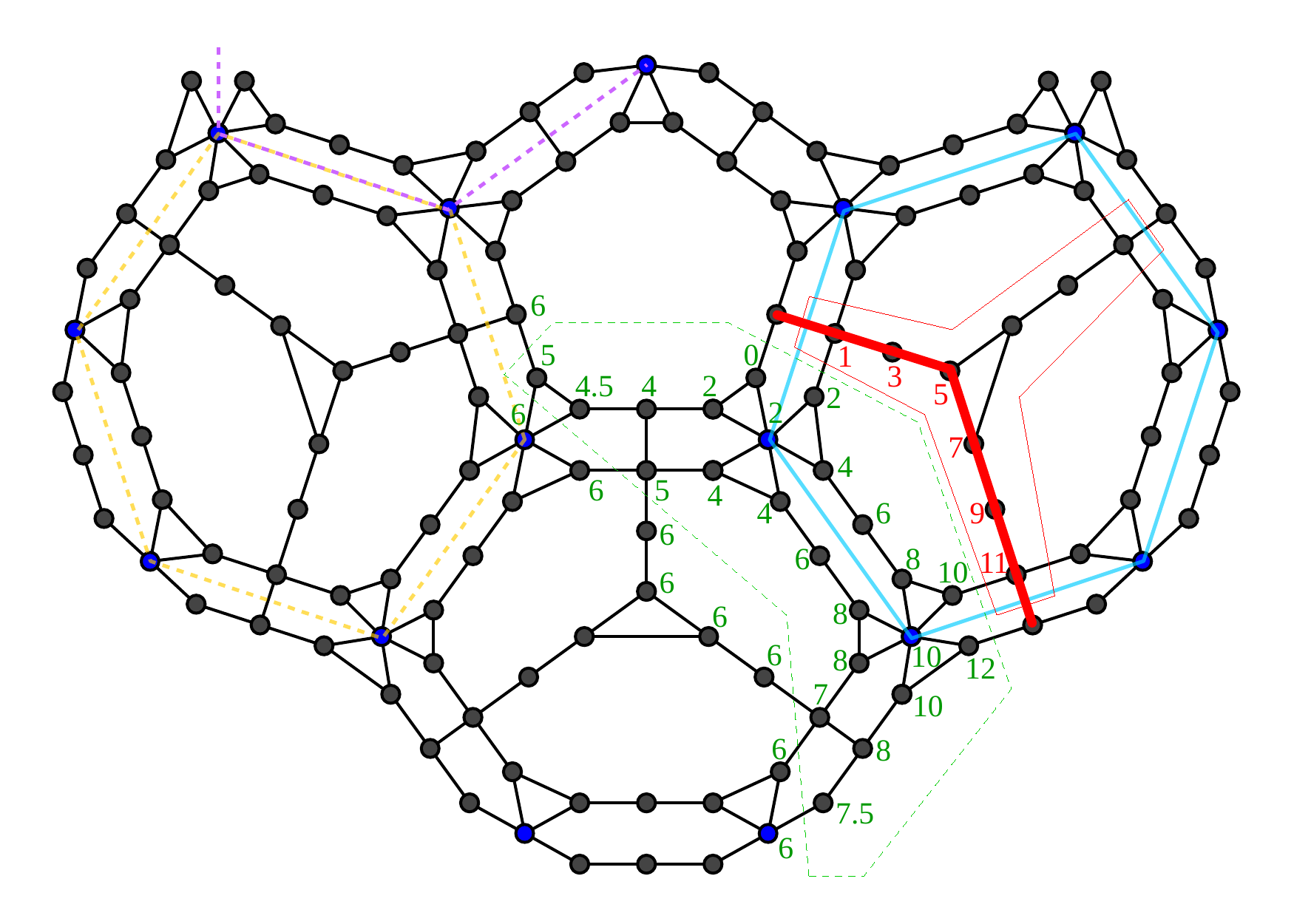}\end{centerbox}
\end{roundedbox}

\end{centerbox}
For a cop to guard the blue hexagon, it must stay within the red-outlined
region, and its position must be as follows.

If the robber is within a \textbf{proximal region} of the hexagon,
shown as the green-dotted region (the other two proximal regions are
positioned symmetrically around the hexagon), it is labelled according
to the number on its vertex as shown, and the cop must be \textbf{\textit{either}}
on the shown red path and labelled according to the number on its
vertex, \textbf{\textit{or}} at the central vertex beside the red
path (next to both the $5$-labelled and $7$-labelled vertices) and
labelled $5$ or $7$, whichever is furthest from the robber's label.
Define the \textbf{deviation} of such a cop to be the (absolute) difference
between the cop's label and the robber's label. (The labelling is
different for a cop guarding another hexagon.) If the robber is not
within any proximal region or its corresponding red path, it is labelled
$6$, and the cop's \textbf{deviation} is defined to be $2d+1$ where
$d$ is its distance from the nearest central vertex. If the robber
is on the red path, the cop's \textbf{deviation} is defined to be
$2d+1$ where $d$ is its distance from the robber. Intuitively, the
cop's deviation captures roughly how far it is from being able to
guard its hexagon, where any deviation of $1$ or less is optimal.

We say that the cop \textbf{guards} this hexagon iff the cop is positioned
as stated above and its deviation is at most $2$. Clearly, if such
a cop can move on every turn, then it can preserve this invariant
and hence prevent the robber from crossing the red path (i.e.~moving
to any of its vertices) without getting caught. In general, we want
to maintain this even with multiple cops each guarding some hexagon.
To do so, we say that a set $S$ of cops \textbf{strongly guard} their
hexagons iff the cops in $S$ guard distinct non-adjacent hexagons
and furthermore at most one cop in $S$ has deviation more than $1$,
and we shall prove the following crucial lemma.
\begin{lem}[Hexagon Guard Lemma]
\label{lem:hex-guard} Take any game state where it is the robber's
turn and some set $S$ of cops strongly guard their hexagons. Then
from that point onwards, the cops can either win or indefinitely preserve
the invariant that the cops in $S$ strongly guard their hexagons
after the cops' turn. If additionally no three cops in $S$ guard
hexagons that are all adjacent to the same hexagon, then in the latter
case the cops can move in such a way that infinitely often on their
turn no cop in $S$ moves.\end{lem}
\begin{proof}
We can assume that the robber does not move next to a cop, otherwise
the cops can immediately win. Thus the robber cannot move onto the
red path for any cop in $S$. Observe that when the robber moves,
the deviation of each cop in $S$ is still well-defined, and there
are at most two relevant cops in $S$, where a cop in $S$ is \textbf{relevant}
iff the robber moved within, into or out of a proximal region of the
hexagon guarded by that cop. So when the robber moves, only relevant
cops in $S$ can have their deviations changed, and by at most $2$,
whereas irrelevant cops in $S$ have their deviation remaining at
$1$. Also, if the new deviation for a relevant cop in $S$ is more
than $1$, then either the robber is now in a proximal region for
that cop, in which case it is possible for the cop to move (in one
step) to a vertex on the red path to adjust its deviation from $d$
to $|d-2|$, or the robber has just left a proximal region, in which
case it is possible for the cop to move to a central vertex to adjust
its deviation in the same way. There are two main cases:
\begin{enumerate}
\item There is at most one cop in $S$ whose deviation $d$ changed. If
now $d\le2$, then the invariant still holds. But if now $d>2$, then
that cop can move to adjust its deviation to $|d-2|$, hence preserving
the invariant since $2<d\le4$ implies $|d-2|\le2$.
\item There are exactly two relevant cops in $S$ with deviations $d,d'$
where initially $d\ge d'$, and their deviations changed by at most
$1$ each. By the strong guarding invariant, initially $d\le2$ and
$d'\le1$, so now $d\le3$ and $d'\le2$. If now $d\le1$, the invariant
already holds. But if now $d>1$, then the cop with deviation $d$
can move to adjust its deviation to $|d-2|$, hence preserving the
invariant since $1<d\le3$ implies $|d-2|\le1$.
\end{enumerate}
We now list the remaining cases for two relevant cops in the following
tables. By symmetry we can assume that one of the relevant cops in
$S$ guards the blue hexagon, and the robber is in the upper half
of the green-dotted region. The first table is for when the other
relevant cop in $S$ guards the purple-dotted hexagon (partly shown),
and the second and third tables are for when the other relevant cop
in $S$ guards the yellow-dotted hexagon. Each case is given on a
separate row, characterized mainly by the change in the labels for
the robber given in the first column (with respect to those two cops).
For each pair of robber label changes, there are only a few cases
in which we need to move a cop to preserve the invariant, and in each
case we can indeed move just one of those two cops to do so, resulting
in the cop label changes given in the second column.
\begin{centerbox}
\begin{tabular}{|c|c|}
\hline 
\textbf{Robber labels} & \textbf{Cop labels}\tabularnewline
\hline 
$(6,6)\to(5,4.5)$ & $(5,7)\to(5,5)$\tabularnewline
$(6,6)\to(5,4.5)$ & $(7,7)\to(7,5)$\tabularnewline
$(5,4.5)\to(6,6)$ & $(5,3)\to(5,5)$\tabularnewline
$(5,4.5)\to(6,6)$ & $(3,5)\to(5,5)$\tabularnewline
\hline 
\end{tabular}\quad{}%
\begin{tabular}{|c|c|}
\hline 
\textbf{Robber labels} & \textbf{Cop labels}\tabularnewline
\hline 
$(6,2)\to(5,0)$ & $(5,3)\to(5,1)$\tabularnewline
$(6,2)\to(5,0)$ & $(7,3)\to(7,1)$\tabularnewline
$(5,0)\to(6,2)$ & $(3,1)\to(5,1)$\tabularnewline
$(6,2)\to(4.5,2)$ & $(7,1)\to(5,1)$\tabularnewline
$(6,2)\to(4.5,2)$ & $(7,3)\to(5,3)$\tabularnewline
$(4.5,2)\to(6,2)$ & $(3,1)\to(5,1)$\tabularnewline
$(4.5,2)\to(6,2)$ & $(3,3)\to(5,3)$\tabularnewline
\hline 
\end{tabular}\quad{}%
\begin{tabular}{|c|c|}
\hline 
\textbf{Robber labels} & \textbf{Cop labels}\tabularnewline
\hline 
$(4.5,2)\to(5,0)$ & $(5,3)\to(5,1)$\tabularnewline
$(4.5,2)\to(5,0)$ & $(3,3)\to(3,1)$\tabularnewline
$(5,0)\to(4.5,2)$ & $(7,1)\to(5,1)$\tabularnewline
$(4.5,2)\to(4,4)$ & $(5,1)\to(5,3)$\tabularnewline
$(4.5,2)\to(4,4)$ & $(3,1)\to(3,3)$\tabularnewline
$(4,4)\to(4.5,2)$ & $(5,5)\to(5,3)$\tabularnewline
$(4,4)\to(4.5,2)$ & $(3,5)\to(3,3)$\tabularnewline
\hline 
\end{tabular}
\end{centerbox}
\clearpage{}

Finally, under the additional assumption that no three cops in $S$
guard hexagons that are all adjacent to the same hexagon, we shall
prove that player $C$ (Cops) can move the cops in $S$ as stipulated
above to preserve the invariant, such that after finitely many turns
the game will reach a state where it is $C$'s turn and $C$ does
not need to move any cop in $S$ to preserve the invariant (i.e. the
invariant is already satisfied). To do so, we shall consider any robber
strategy where $C$ always has to move some cop in $S$ to preserve
the invariant, and show that it is impossible.

Each robber move must change the cop deviations, so we can assume
that one of the relevant cops $x$ in $S$ guards the blue hexagon,
and it is not hard to verify that:
\begin{enumerate}
\item The robber cannot move to an adjacent vertex with the same robber
labels (including moving to or from a central vertex or along one
of the broken edges in the below-left diagram).
\item The robber cannot move to a vertex with the same robber labels as
one move ago (including moving backwards along the same edge that
it used in the previous turn), otherwise the invariant would still
be satisfied without any cop moving. This entails checking each of
the above cases one by one:

\begin{enumerate}
\item First main case: Only one cop in $S$ has deviation changed after
the (previous) robber move. With respect to that cop, let $c,r$ be
the labels for the cop and robber respectively after that move, and
$r'$ be the robber label before that move. It must be that $|c-r|>2$
to make that cop move, hence by symmetry we can assume $c+2<r$, and
the new cop label is $c+2$. Trivially $r\le r'+2$ (see the diagram
in \ref{sub:hex-guard}), yielding $c<r'$. And $r'\le c+2$ since
that cop was guarding its hexagon before the robber move. Thus $c<r'\le c+2$
and hence $|(c+2)-r'|\le2$, so if the robber returns to the previous
labels, then no cop needs to move.
\item Second main case: Exactly two cops in $S$ have deviations $d,d'$
changed after the (previous) robber move, each by at most $1$, where
initially $d\ge d'$ and so $d'\le1$. It must be that after that
move $d>1$ to make the corresponding cop move. With respect to that
cop, let $c,r$ be the labels for the cop and robber respectively
after that robber move. Then $|c-r|>1$, so by symmetry we can assume
$c+1<r$, and the new cop label is $c+2$. It cannot be that $r'<c$,
otherwise $r\le r'+2\le c+1.5\le r$, which forces $r=r'+2=c+1.5$
and implies that $r,r'$ are distinct non-integers, which is impossible
(see the diagram in \ref{sub:hex-guard}). Thus $c\le r'\le c+2$
and hence $|(c+2)-r'|\le2$, so the robber must not return to the
previous labels, otherwise we once again have $d\le2$ and $d'\le1$
and hence no cop needs to move.
\item Remaining cases: We can easily check that, in each row of the above
tables, the new cop labels satisfy the desired invariant with the
old robber labels.
\end{enumerate}
\end{enumerate}
From these we can infer that the robber also cannot move to the vertices
erased in the below-right diagram. For convenience, we also mark three
of the $6$-labels as $6a,6b,6c$ to distinguish those vertices for
easy reference later.
\begin{centerbox}
\begin{roundedbox}
\begin{centerbox}
\includegraphics[bb=140bp 14bp 490bp 343bp,clip,scale=0.91]{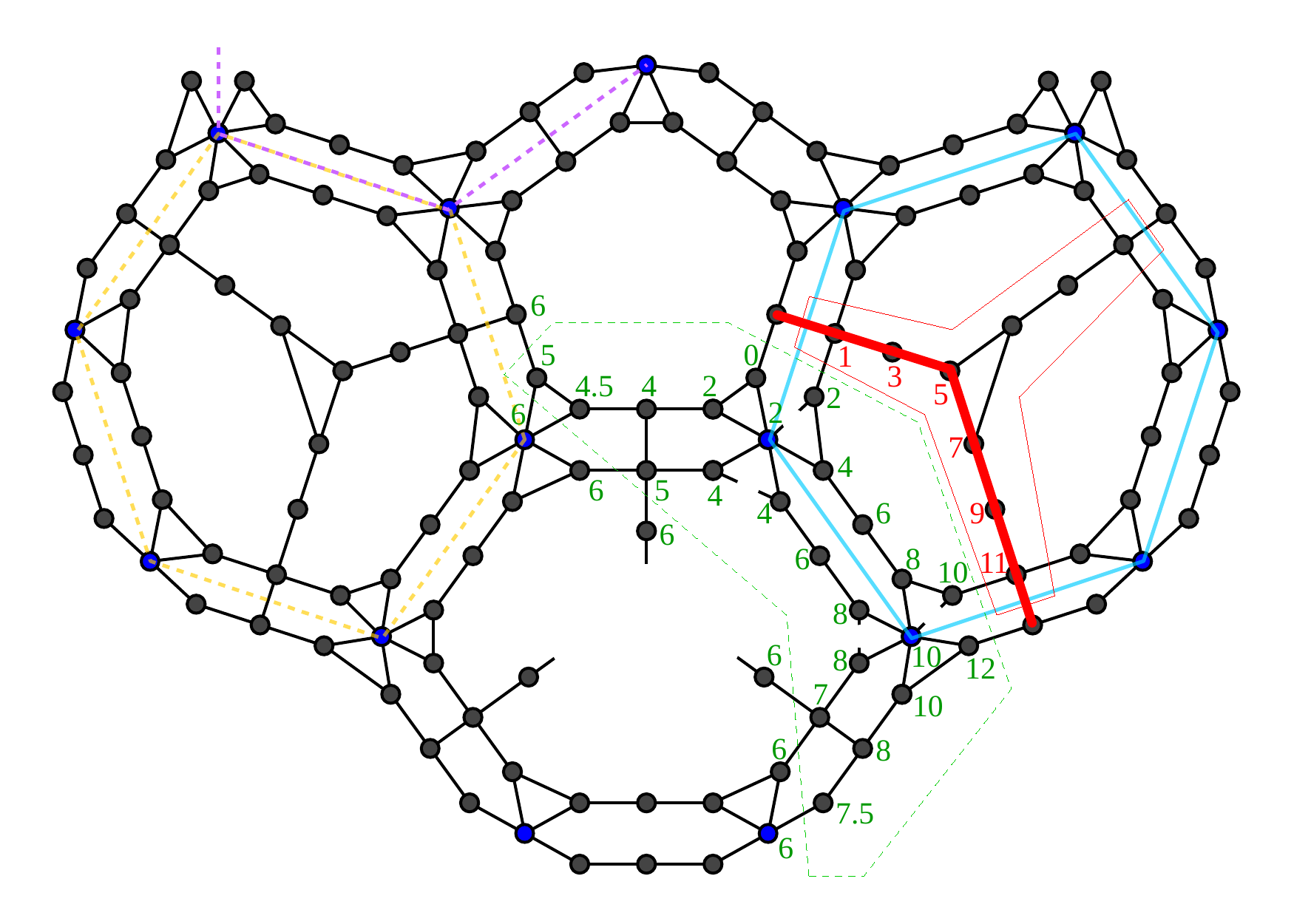}\quad{}\includegraphics[bb=140bp 14bp 490bp 343bp,clip,scale=0.91]{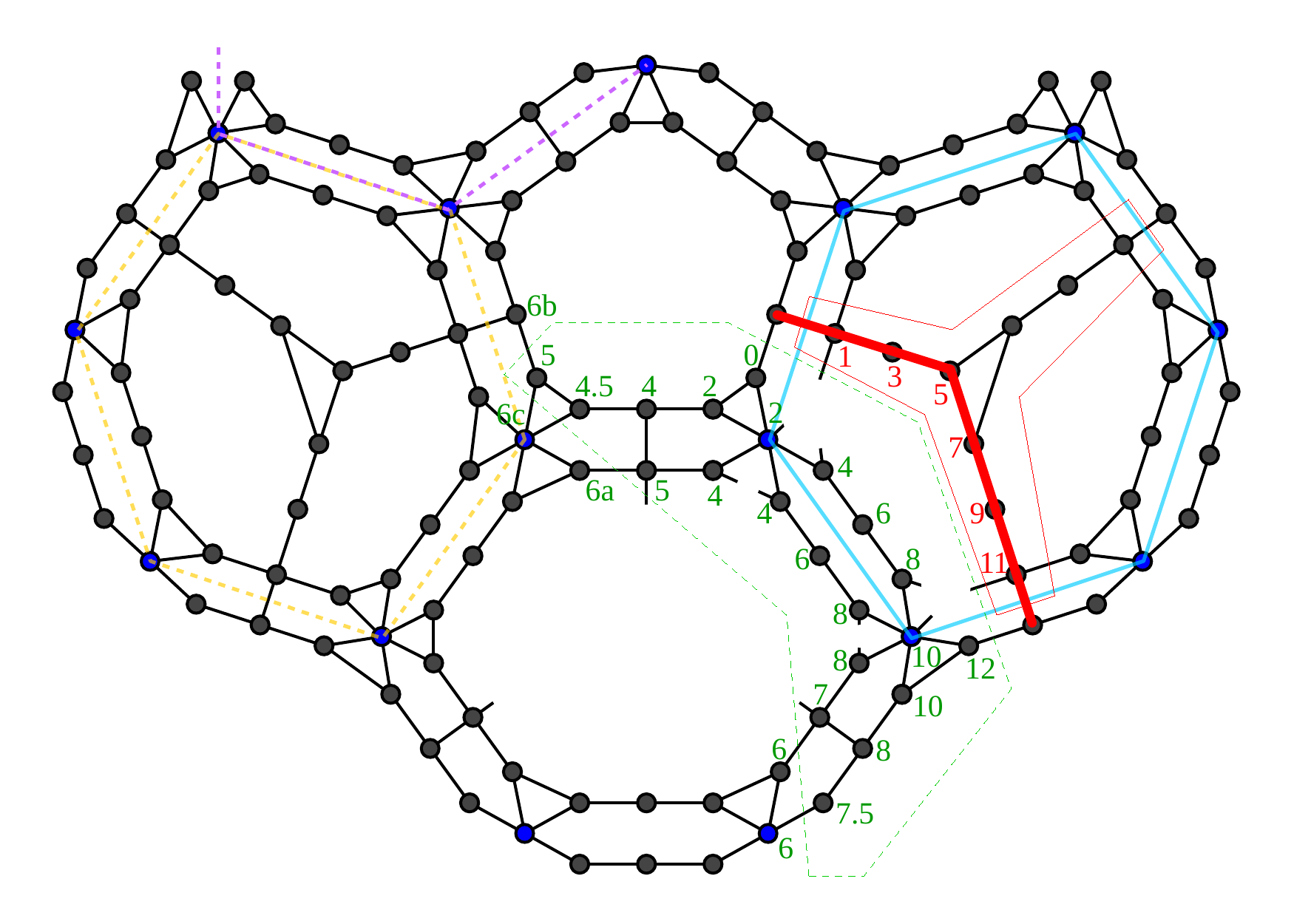}\end{centerbox}
\end{roundedbox}

\end{centerbox}
\clearpage{}

It remains to analyze all ways the robber can enter the proximal region
for $x$ (classified by the robber label sequence):
\begin{enumerate}
\item The robber moves $6a\to5$. For the invariant to be violated after
that, there must be a cop $y$ in $S$ guarding the yellow hexagon,
and the robber label for $y$ changes $4\to5$, and the cop labels
for $(x,y)$ must be initially $(7,3)$. But the cops can move $(7,3)\to(5,3)$,
after which the robber cannot move $(5,5)\to(4,4)$ (vertically) otherwise
the cops do not need to move. So the robber must continue $(5,5)\to(4,6)$
(horizontally), and the cops continue $(5,3)\to(5,5)$. After that,
the robber is forced to continue $(4,6)\to(2,6)\to(0,5)\to(2,4.5)$,
and the cops continue $(5,5)\to(3,5)\to(1,5)$ and then do not need
to move (on the next turn).
\item The robber moves $6b\to5$. For the invariant to be violated after
that, there must be a cop $y$ in $S$ guarding the purple hexagon,
and the robber label for $y$ changes $4\to4.5$, and the cop labels
for $(x,y)$ must be initially $(7,3)$. But the cops can move $(7,3)\to(5,3)$,
after which the robber cannot move $(5,4.5)\to(4.5,5)$, and so must
continue $(5,4.5)\to(6,6)\to(4.5,5)$, and the cops continue $(5,3)\to(5,5)$,
and then do not need to move.
\item The robber moves $6c\to5$. Then for the invariant to be violated
after that move, there must be a cop $y$ in $S$ guarding either
the purple hexagon or the yellow hexagon.

\begin{enumerate}
\item \label{enu:halt} If $y$ is guarding the purple hexagon, then the
robber label for $y$ changes $6\to4.5$, and the cop labels for $(x,y)$
must be initially $(5,7)$ or $(7,7)$. $y$ moves $7\to5$, after
which the robber cannot move $(5,4.5)\to(6,4)$, and so must continue
$(5,4.5)\to(4.5,5)$. It must be that the cop labels are now $(7,5)$
to force them to continue $(7,5)\to(5,5)$. But after that the robber
must continue $(4.5,5)\to(4,6)$ or $(4.5,5)\to(6,6)$, and then the
cops do not need to move.
\item \label{enu:run1} If $y$ is guarding the yellow hexagon, then the
robber label for $y$ changes $2\to0$, and the robber must move $(6,2)\to(5,0)\to(4.5,2)$,
and the cop labels for $(x,y)$ must be initially $(5,3)$ or $(7,3)$.
$y$ moves $3\to1$, after which the cop labels must be $(7,1)$ to
force them to continue $(7,1)\to(5,1)$. After that, the robber cannot
move $(4.5,2)\to(6,2)$, and so must continue $(4.5,2)\to(4,4)$,
and the cops continue $(5,1)\to(5,3)$. Again, the robber cannot move
$(4,4)\to(5,5)$, but must continue $(4,4)\to(2,4.5)\to(0,5)\to(2,6)$,
and the cops continue $(5,3)\to(3,3)\to(1,3)\to(1,5)$. Again, the
robber cannot move $(2,6)\to(2,4.5)$, nor $(2,6)\to(4,6)\to(5,5)$
otherwise the cops continue $(1,5)\to(3,5)$ and then do not need
to move. Hence the robber must continue on the path $2\to4\to6\to8\to10$
(labels for $x$). On the next move after that, if the robber does
not enter a new proximal region, it must move $10\to12$ to force
$x$ to move $9\to11$, but on the subsequent turn no cop needs to
move. Therefore there must be another cop in $S$ guarding the hexagon
just below the bottommost hexagon in the diagram, and the robber must
enter its proximal region. 
\end{enumerate}
\item The robber moves $6c\to4.5$. There must be a cop $y$ in $S$ guarding
either the purple hexagon or the yellow hexagon, otherwise the cop
label for $x$ must be initially $7$, and $x$ moves $7\to5$, after
which the robber must continue $4.5\to4$ or $4.5\to5$ so no cop
needs to move.

\begin{enumerate}
\item If $y$ is guarding the purple hexagon, then by symmetry the situation
is exactly as in \case\ref{enu:halt}.
\item \label{enu:run2} If $y$ is guarding the yellow hexagon, then the
robber label for $y$ remains at $2$, so the cop label for $x$ must
be initially $7$ and $x$ moves $7\to5$. At this point, the cop
labels for $(x,y)$ cannot be $(5,3)$, otherwise the robber cannot
move $(4.5,2)\to(4,4)$ or $(4.5,2)\to(6,2)$, and so must continue
$(4.5,2)\to(5,0)\to(6,2)$, but the cops continue $(5,3)\to(5,1)$
and then do not need to move. Therefore the cop labels must be $(5,1)$,
so the robber cannot move $(4.5,2)\to(5,0)$ or $(4.5,2)\to(6,2)$,
and must continue $(4.5,2)\to(4,4)$, and the cops continue $(5,1)\to(5,3)$.
After that, as in the later half of \case\ref{enu:run1}, the robber
must continue on the path $4\to2\to0\to2\to4\to6\to8\to10$ of robber
labels for $x$, and then enter a proximal region of a hexagon that
is guarded by a cop in $S$ and just below the bottommost hexagon
in the diagram.
\end{enumerate}
\end{enumerate}
Therefore the robber must indefinitely repeat \case\ref{enu:run1}
or \case\ref{enu:run2}. But this is impossible, because it requires
three cops in $S$ guarding hexagons that are all adjacent to the
same hexagon (the bottommost one in the diagram).\end{proof}
\begin{rem*}
Incidentally, if three cops in $S$ guard hexagons that are all adjacent
to the same hexagon, then even if the robber is `confined' inside
the region around the central hexagon, the robber can indefinitely
repeat the path $5\to4.5\to4\to2\to0\to2\to4\to6\to8\to10\to12$ to
force those three cops to keep moving, which implies that the cops
cannot catch the robber unless they break out of this guarding pattern!
This is one reason it seems difficult to ascertain whether or not
$cop_{1}(\g)=4$.
\end{rem*}
The next lemma captures how we can \textbf{expand} strong guarding
of some hexagons to an extra hexagon (using an extra cop), while still
strongly guarding the original hexagons. Consequently, once the cops
have confined the robber to a region by strongly guarding some hexagons,
then the cops can keep the robber confined to that region while moving
an extra cop to strongly guard yet another hexagon, to confine the
robber even further.
\begin{lem}[Guard Expansion Lemma]
\label{lem:guard-expand} Take any set $H$ of distinct non-adjacent
hexagons, no three of which are adjacent to the same hexagon. And
take any game state, where some set $S$ of cops strongly guard all
the hexagons in $H$ except some hexagon $L$, and there is another
cop $x$ not in $S$. Then the cops can move in such a way that the
cops in $S$ still always strongly guard their hexagons (after their
turn) and yet eventually the cops in $S\cup\{x\}$ strongly guard
all the hexagons in $H$.\end{lem}
\begin{proof}
We start by gradually moving $x$ to a central vertex of $L$ while
maintaining the invariant that the cops in $S$ strongly guard their
hexagons (after their turn), by the \nameref{lem:hex-guard}. After
that, if the robber is not in a proximal region of $L$, then the
cops in $S\cup\{x\}$ strongly guard all the hexagons in $H$ and
we are done. But if the robber is in a proximal region of $L$, then
there is a vertex $v$ on the corresponding red path (see the diagram
in \ref{sub:hex-guard}) such that $S\cup\{x\}$ would strongly guard
all the hexagons in $H$ if $x$ is at $v$. Place a guide $x'$ at
$v$. From then on, after each robber's turn, we perform the following
steps:
\begin{enumerate}
\item If $S\cup\{x'\}$ does not strongly guard their hexagons (treating
$x'$ as an actual cop), move one cop/guide in $S\cup\{x'\}$ so that
they (again) strongly guard their hexagons, by the \nameref{lem:hex-guard}.
Otherwise move nothing.
\item If in step~1 we moved $x'$ or nothing at all, then move $x$ towards
$x'$ (if it is not already at the same vertex).
\end{enumerate}
Note that this yields valid moves because on each cops' turn we move
only one cop. Also, the distance from $x$ to $x'$ (measured after
the cops' turn) never increases (since $x$ is moved whenever $x'$
is moved), and it decreases repeatedly until it is zero because infinitely
often no cop/guide is moved in step~1, by the \nameref{lem:hex-guard}
again. Thus eventually $x$ is at the same vertex as $x'$ and hence
the cops in $S\cup\{x\}$ strongly guard all the hexagons in $H$.
\end{proof}

\subsection{The Winning Strategy}

\scalegraphics{.63}

\begin{minipage}[t]{0.56\columnwidth}%
Now we are ready to present the winning strategy for $7$ cops. Identify
$8$ of the hexagons underlying $\g$ whose centres form a cube, and
divide them into \textbf{$6$ red} and $2$ \textbf{yellow} hexagons,
where the yellow hexagons are opposite the centre of the cube (as
in the diagram on the right). Note that no three of the $6$ red hexagons
are adjacent to the same hexagon.

Let $S$ be a set of $6$ of the cops. At the start, place each cop
in $S$ at a central vertex of a different red hexagon, and place
the $7$th cop anywhere. Once the opponent has placed the robber,
let $T$ be the set of cops in $S$ with deviation at most $1$, and
note that the cops in $T$ strongly guard their hexagons. The idea
is roughly to expand strong guarding from those hexagons to all red
hexagons to confine the robber to `half of the cube', and then keep
the robber there while expanding also to the yellow hexagon in that
`half' to further confine the robber to a `square of the cube'.
After that, we use the corresponding $4$ cops to continue strongly
guarding that `square' while moving $2$ other cops to divide the
confinement region in half, and then gradually reduce it further.%
\end{minipage}\hfill{}%
\begin{minipage}[t]{0.42\columnwidth}%
\begin{centerbox}
\begin{roundedbox}
\begin{centerbox}
\includegraphics[bb=0bp 0bp 302bp 302bp]{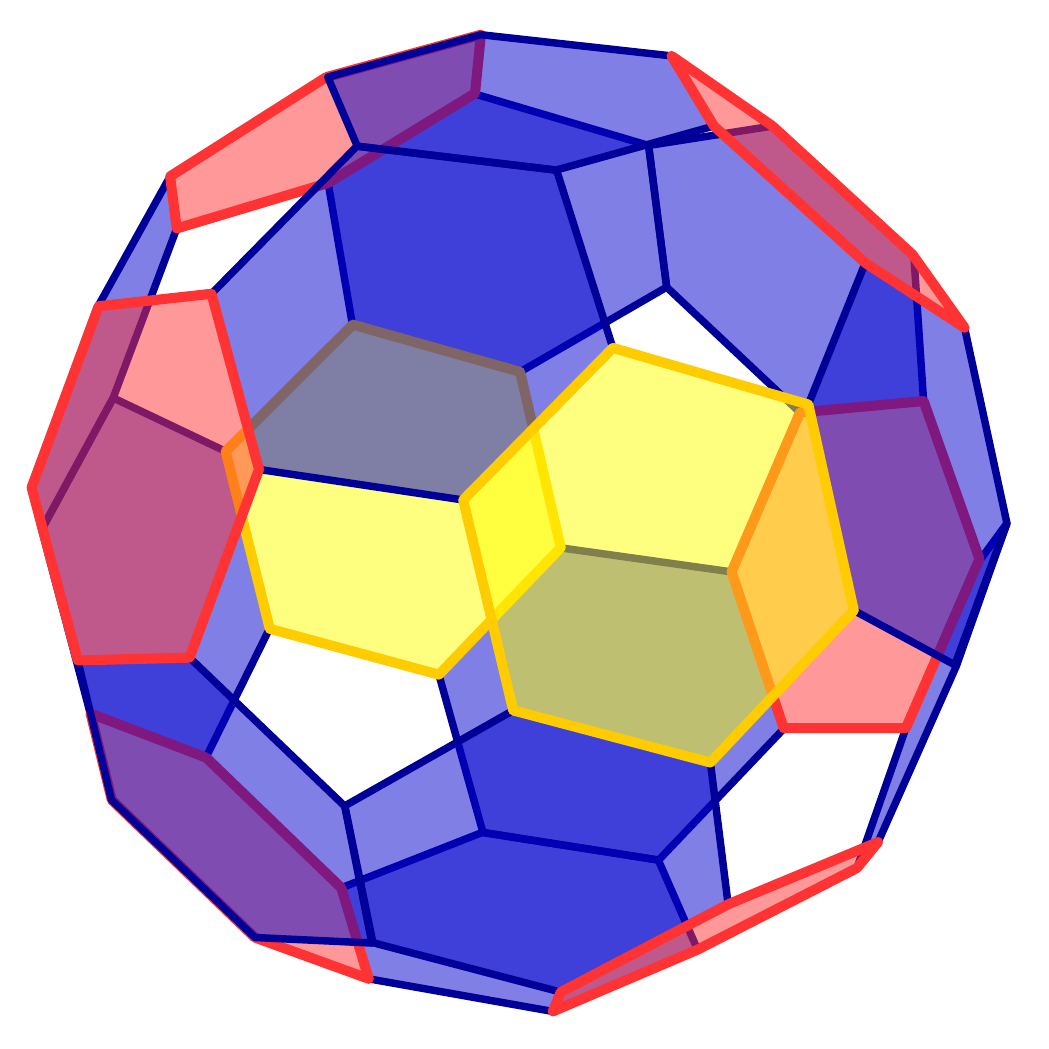}\end{centerbox}
\end{roundedbox}
\end{centerbox}
\end{minipage}

We shall now go into the details of how to move the cops.

\begin{minipage}[t]{0.56\columnwidth}%

\subsubsection{Movement Phase 1}

Move the cops in $S$ to eventually strongly guard the red hexagons,
by applying the \nameref{lem:guard-expand} (\ref{lem:guard-expand})
to expand strong guarding from the hexagons guarded by the cops in
$T$ to all the red hexagons, one hexagon at a time. After this, the
robber will be confined to one of the two possible `halves of the
cube' on either `side' of the `ring' of red hexagons, where one
`side' is represented on the right by the coloured hexagons. (Of
course, the robber is confined to only one `side' of each red hexagon.)

Note that, during this phase, we do not care if the robber `escapes'
past any of the vertices that the cops are eventually supposed to
guard. All that matters is that after finitely many moves, these $6$
cops strongly guard their red hexagons and hence the robber will be
confined to one `side' of that `ring' of hexagons.%
\end{minipage}\hfill{}%
\begin{minipage}[t]{0.42\columnwidth}%
\begin{centerbox}
\begin{roundedbox}
\begin{centerbox}
\includegraphics[bb=0bp 0bp 302bp 302bp]{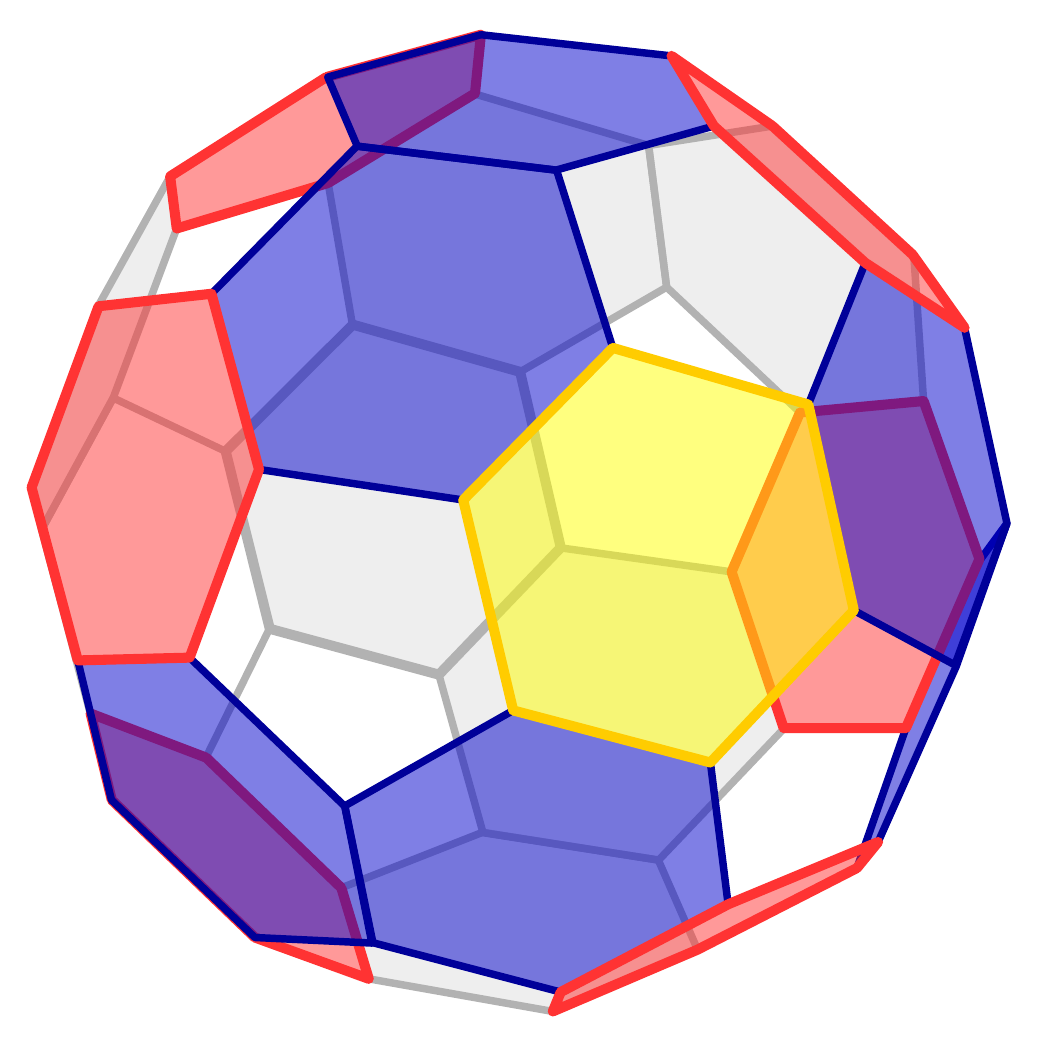}\end{centerbox}
\end{roundedbox}
\end{centerbox}
\end{minipage}

\begin{minipage}[t]{0.56\columnwidth}%

\subsubsection{Movement Phase 2}

Keep the robber confined to its current `half of the cube', while
moving the remaining $7$th cop to `expand' strong guarding to the
yellow hexagon in that `half', again by the \nameref{lem:guard-expand}
(\ref{lem:guard-expand}). After this, the robber will be confined
to one of three possible `squares of the cube', represented on the
right by the coloured hexagons (the other cases are symmetric).

Now keep the robber confined to its current `square', using the
$4$ cops that had been strongly guarding the red/yellow hexagons,
while moving $2$ of the other cops to the key vertices shared by
blue hexagons (as indicated by the pink circles), by the \nameref{lem:hex-guard}
(\ref{lem:hex-guard}). After this, the robber will be confined to
`half of that square' or between the pink-circled key vertices.%
\end{minipage}\hfill{}%
\begin{minipage}[t]{0.42\columnwidth}%
\begin{centerbox}
\begin{roundedbox}
\begin{centerbox}
\includegraphics[bb=0bp 0bp 302bp 302bp]{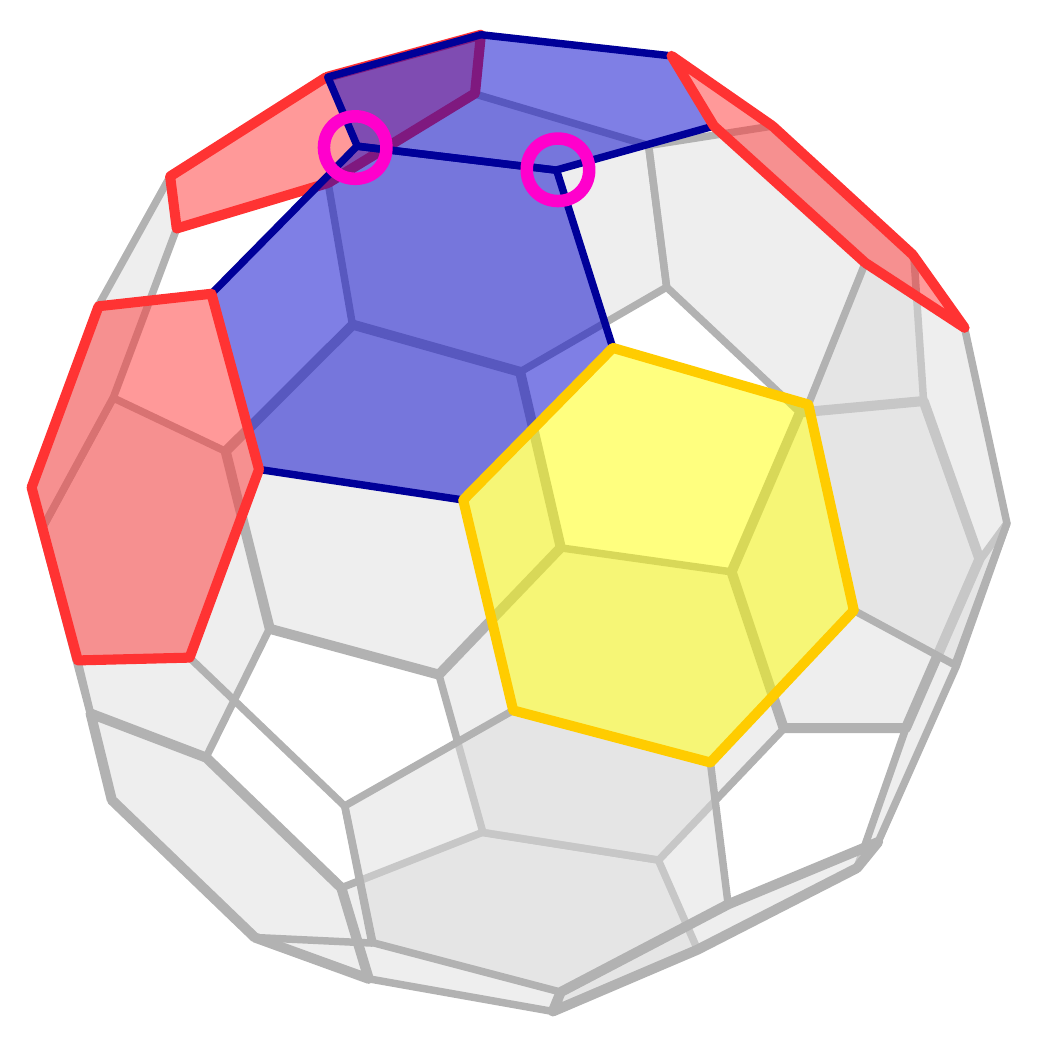}\end{centerbox}
\end{roundedbox}
\end{centerbox}
\end{minipage}

\scalegraphics{.455}

\begin{minipage}[t]{0.49\columnwidth}%

\subsubsection{Movement Phase 3}

More precisely, we can assume that the robber never moves next to
a cop, so the robber is now confined to one of the green-dotted regions
in the diagram on the right, with $2$ \textbf{red }cops strongly
guarding the red hexagons (one on each red path) and $2$ \textbf{pink}
cops $p,q$ shown as solid pink circles with $p$ on the left. If
the robber is in-between the pink cops, we can trivially move a third
cop to catch the robber. Otherwise we maintain strong guarding of
the red hexagons, while moving a $5$th cop to the pink-circled vertex,
by the \nameref{lem:hex-guard} (\ref{lem:hex-guard}).%
\end{minipage}\hfill{}%
\begin{minipage}[t]{0.49\columnwidth}%
\begin{centerbox}
\begin{roundedbox}
\begin{centerbox}
\includegraphics[bb=0bp 0bp 504bp 360bp]{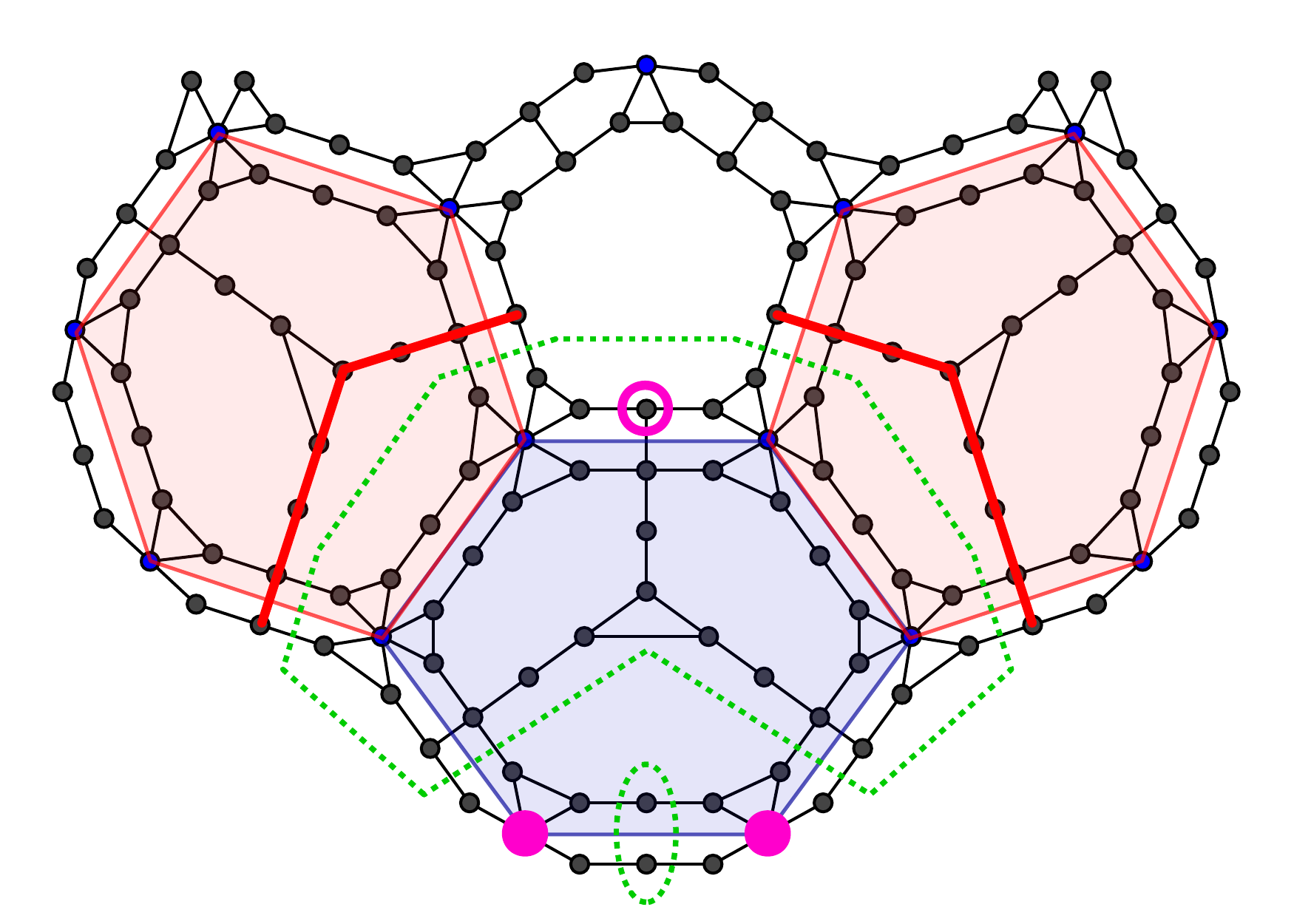}\end{centerbox}
\end{roundedbox}
\end{centerbox}
\end{minipage}

\begin{minipage}[t]{0.49\columnwidth}%
Maintain strong guarding of the red hexagons, until we do not have
to move any red cop, by the \nameref{lem:hex-guard} (\ref{lem:hex-guard}).

By left-right symmetry we can assume that the robber is at this point
in the green-dotted region as shown on the right. Move the pink cop
$q$ one step along the pink path, and continue gradually moving it
along the path while maintaining guarding of the left red hexagon.
Observe that after that first step along that path, there is no need
to guard the right red hexagon anymore, and that on each subsequent
step, the robber is confined to a smaller region.%
\end{minipage}\hfill{}%
\begin{minipage}[t]{0.49\columnwidth}%
\begin{centerbox}
\begin{roundedbox}
\begin{centerbox}
\includegraphics[bb=0bp 0bp 504bp 360bp]{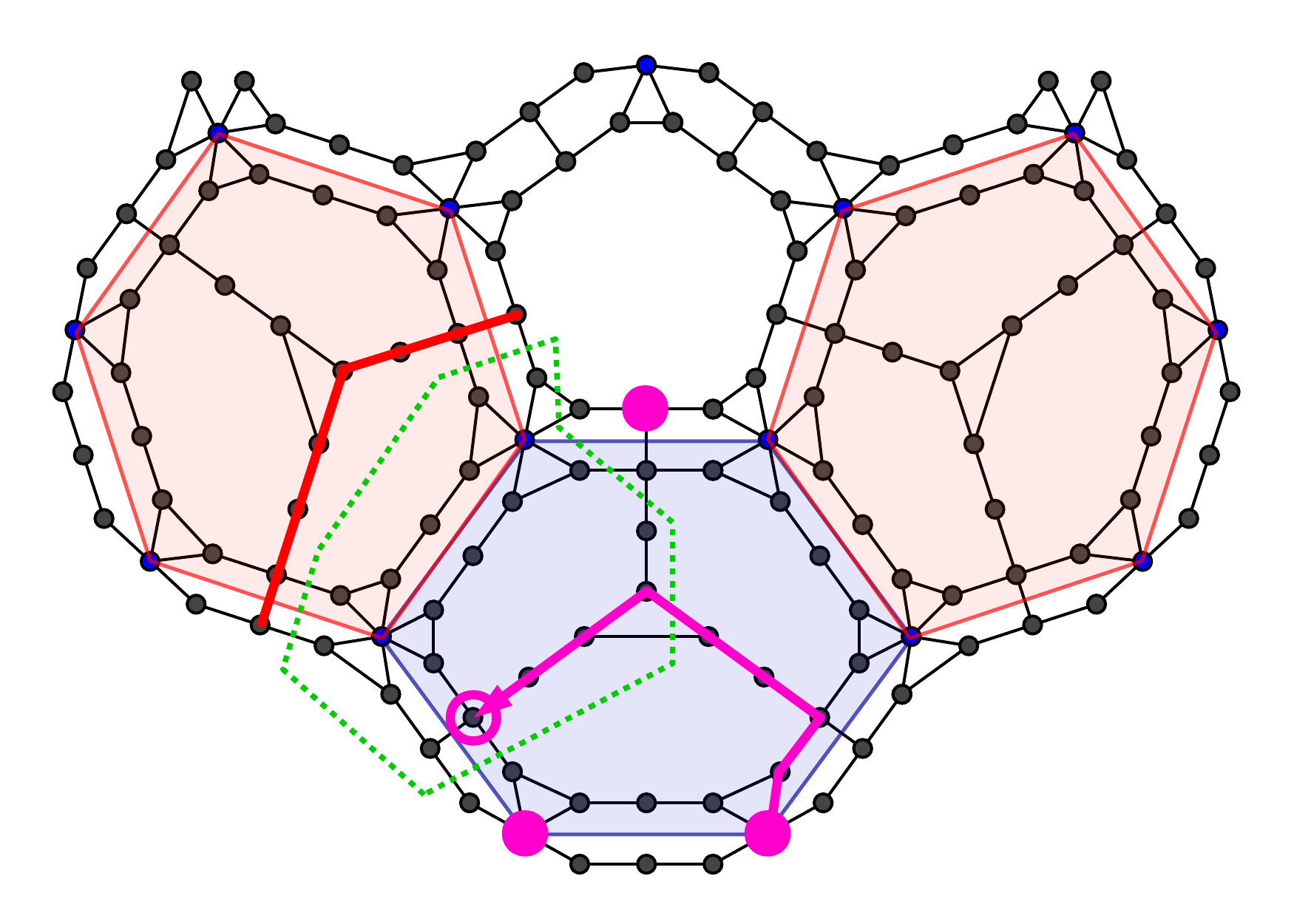}\end{centerbox}
\end{roundedbox}
\end{centerbox}
\end{minipage}

\begin{minipage}[t]{0.49\columnwidth}%
Finally, as shown on the right, gradually move the pink cop $p$ along
the given path (labelled ``1''), followed by the pink cop $q$ along
the given path (labelled ``2''), all the while maintaining guarding
of the left red hexagon, and the robber will be caught.%
\end{minipage}\hfill{}%
\begin{minipage}[t]{0.49\columnwidth}%
\begin{centerbox}
\begin{roundedbox}
\begin{centerbox}
\includegraphics[bb=0bp 0bp 504bp 360bp]{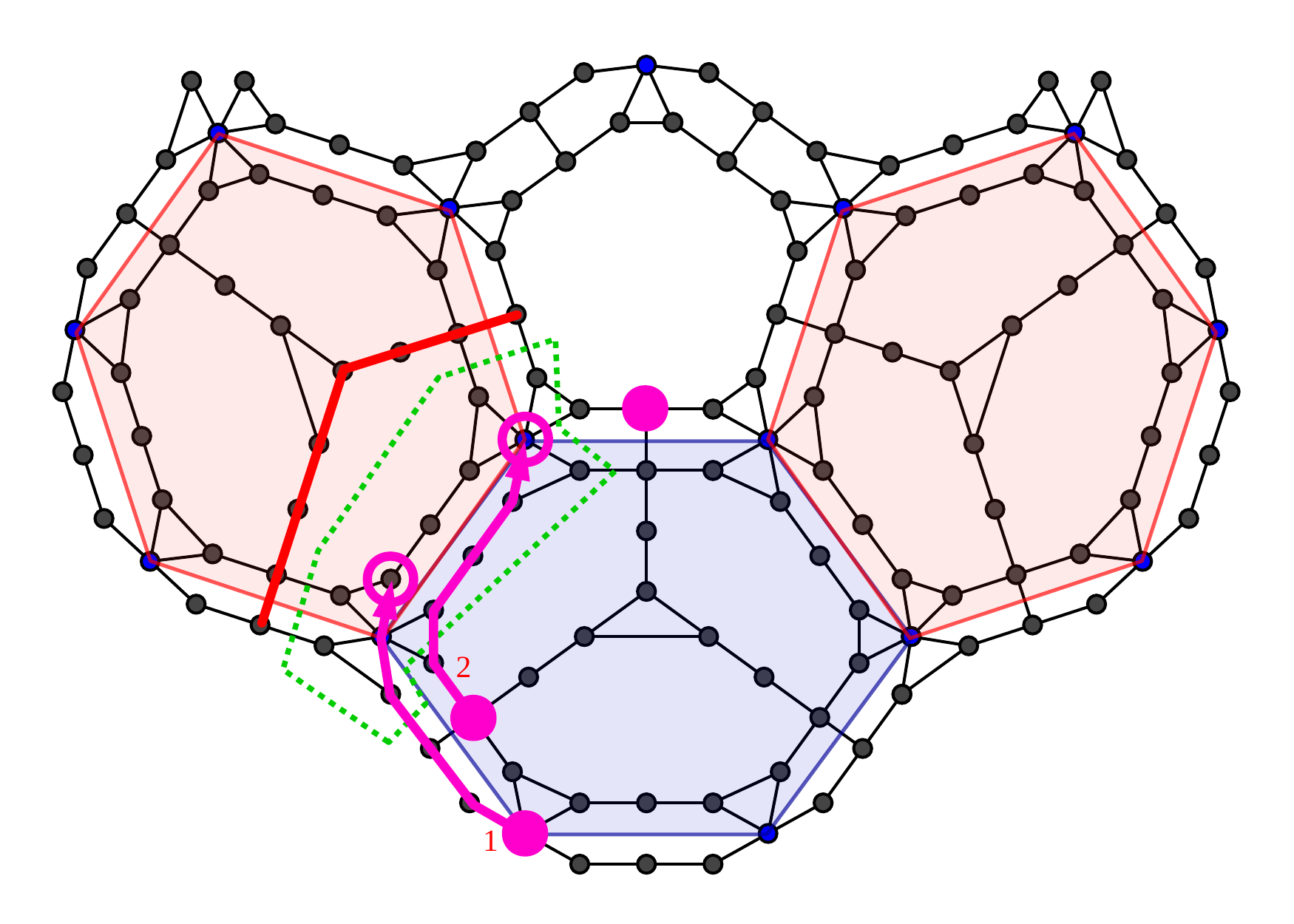}\end{centerbox}
\end{roundedbox}
\end{centerbox}
\end{minipage}

\section{1-Cop-Move Number 4 on a Small Graph}

Although it seems difficult to find a strategy for fewer cops to win
on $\g$, we can easily show that there is some connected (planar)
subgraph $\h$ of $\g$ such that $cop_{1}(\h)=4$. This follows readily
from the following lemma.
\begin{lem}[Vertex Guard Lemma]
\label{lem:vertex-guard} Take any graph $G$ and a vertex $v$ in
$G$. Then $cop_{1}(G)\le cop_{1}(G-v)+1$.\end{lem}
\begin{proof}
Let $c=cop_{1}(G-v)$. Then we can use $1$ cop to guard $v$ by staying
there without moving, forcing the robber to never move to $v$, and
hence we can use $c$ other cops to catch the robber on the graph
$G-v$.\end{proof}
\begin{thm}
There is a connected planar graph $G$ with at most $720$ vertices
such that $cop_{1}(G)=4$.\end{thm}
\begin{proof}
Let $G_{0}=\g$ and let $n$ be the number of vertices in $G_{0}$.
For each $k\in[1..n-1]$ let $G_{k}=G_{k-1}-v_{k}$ where $v_{k}$
is a vertex in $G_{k-1}$ that is not a cut vertex (i.e.~$G_{k}$
is still connected). Clearly $cop_{1}(G_{n-1})=1\le4$, so there is
some minimum $m\in[0..n-1]$ such that $cop_{1}(G_{m})\le4$, If $m=0$,
then $cop_{1}(G_{m})=cop_{1}(\g)\ge4$. Otherwise if $m>0$, then
$cop_{1}(G_{m-1})>4$ and hence $cop_{1}(G_{m})\ge4$ by the \nameref{lem:vertex-guard}.
In either case, $cop_{1}(G_{m})=4$ and $G_{m}$ is a connected (planar)
subgraph of $\g$.
\end{proof}

\section{Open Questions}

It is not clear what the true value of $cop_{1}(\g)$ is, and it would
be very interesting if it was more than $4$, because then the robber's
winning strategy against $4$ cops would have to be very different
from the one given in this paper against $3$ cops. One also hopes
that we will eventually find an explicit simpler and smaller finite
connected planar graph with $1$-cop-move number exactly $4$, and
get a better understanding of whether finite connected planar graphs
have bounded $1$-cop-move number or not.

\section*{}

\phantomsection\bibliographystyle{plain}
\addcontentsline{toc}{section}{\refname}\bibliography{OCMG}

\end{document}